\documentclass[american,aps,prx,reprint,floatfix,nofootinbib,superscriptaddress,longbibliography]{revtex4-2}

\usepackage{graphicx}
\usepackage{dcolumn}
\usepackage{bm}
\usepackage{amsmath}
\usepackage{amssymb}
\usepackage{courier}
\usepackage{cprotect}
\usepackage{amsthm}
\usepackage{enumerate}
\usepackage{physics}
\usepackage{soul}
\usepackage{bbm}
\usepackage[dvipsnames]{xcolor}
\usepackage[hyperfootnotes=false]{hyperref}
\usepackage{caption}
\usepackage[capitalise]{cleveref}
\usepackage{float}

\newtheorem{theorem}{Theorem}[section]
\newcommand{\fp}{\mathcal{F}_{\text{proc}}}

\newif\ifshowcomments
\showcommentsfalse
\usepackage{xcolor}

\begin{document}

\title{Learning error suppression strategies for dynamic quantum circuits}
\author{Christopher Tong}
\email{lctong@mit.edu}
 \altaffiliation{current: Princeton University}
 \affiliation{\vspace{-1pt}Department of Physics, Massachusetts Institute of Technology, Cambridge, MA 02139, USA}
 \affiliation{IBM Quantum, IBM Thomas J. Watson Research Center, Yorktown Heights, NY 10598, USA}
\author{Liran Shirizly}
 \affiliation{IBM Quantum, IBM Thomas J. Watson Research Center, Yorktown Heights, NY 10598, USA}
\author{Edward H. Chen}
 \altaffiliation{current: Microsoft Discovery \& Quantum}
 \affiliation{IBM Quantum, Research Triangle Park, NC 27709, USA}
\author{Derek S. Wang}
\altaffiliation{current: independent researcher}
 \affiliation{IBM Quantum, IBM Thomas J. Watson Research Center, Yorktown Heights, NY 10598, USA}
\author{Bibek Pokharel}
 \affiliation{IBM Quantum, IBM Thomas J. Watson Research Center, Yorktown Heights, NY 10598, USA}
\date{\today}

\begin{abstract}
Dynamic quantum circuits integrate unitary evolution with mid-circuit measurement and feedforward, enabling conditional operations essential for efficient quantum algorithms and foundational for fault-tolerant quantum computation. However, such operations introduce measurement-induced errors and control constraints. Here, we introduce an empirical learning framework that optimizes dynamical decoupling (DD) sequences for dynamic circuits at the level of circuit subintervals and qubit subregisters. Applying empirically learned DD sequences, we achieve a two- to three-fold reduction in dynamic circuit error rates, as measured via randomized benchmarking, relative to undecoupled circuits. We apply the learned strategies to the dynamic circuit implementation of the quantum Fourier transform with measurement (QFT+M), demonstrating nontrivial process fidelity on connected chains of up to 20 qubits. Applying the resulting enhancement, we perform a high signal-to-noise QFT immediately following the preparation of a 10-qubit entangled state. Our results demonstrate that empirically optimized DD sequences systematically outperform canonical, theoretically derived DD sequences for dynamic circuits. We establish empirical DD learning as an efficient approach for error suppression in dynamic quantum circuits, with direct relevance to applications requiring measurement and feedback such as quantum error correction.
\end{abstract}

\maketitle
\section{Introduction}
\label{sec:intro}
Dynamic quantum circuits, where unitary gates are interleaved with mid-circuit measurements (MCMs), feedforward (FF) of measurement results, and classical control~\cite{Joz05, Cor21, Pin21, Mos23}, define a hybrid execution model for quantum processors. These non-unitary circuits enable mid-circuit error detection and correction, which are key to realizing fault-tolerant quantum computation~\cite{Sho97, Ter15}, novel quantum dynamics~\cite{Ski19,Fis23,Iva24}, and efficient state preparation schemes~\cite{Smi23,Ala24,Pue24,Smi24,Bhu25, Jeo25,zi2025constantdepthquantumcircuitsarbitrary}. To these ends, dynamic circuits are now routinely implemented on quantum processors across platforms, establishing measurement-plus-unitary computation as a framework for quantum information processing~\cite{Lut22, Koh23,Tan23, Fos23, Bau23, Dec23,Sin23, Bau24, Mal24, Buh24, Vaz24, Che25,Goo25, Pok25}.

However, the expanded computational capabilities afforded by dynamic circuits come with new sources of error absent in unitary-only execution. In addition to readout errors, dynamic circuits are susceptible to decoherence of idle qubits during MCMs, which typically have durations much longer than those of unitary gates~\cite{Cra16,Wal17,Pin21,Gra23,Mau23}, as well as to coherent single-qubit and crosstalk errors generated by the measurement process~\cite{Gov23,Shi24}. These mechanisms motivate the development of error suppression, detection, and mitigation strategies tailored for dynamic circuits~\cite{Niu24,Gup24,Shi25,santos2025driftresilientmidcircuitmeasurementstate,shirgure2026characterizingbenchmarkingdynamicquantum,shirgure2026errormitigationdynamiccircuits}.

Dynamical decoupling (DD) applies sequences of unitary control pulses to qubits to provide low-overhead suppression of coherent errors, effectively averaging away unwanted system–bath interactions~\cite{Vio98,Vio99,Vit99,Zan99,Vio03, Gen17, Zho23,Shi24dissipative, Bro24}. DD has been well-established in error-suppression across a wide range of quantum algorithms and physical platforms~\cite{Bie09, Wan12, Duj09, Fur16, Del10, Ezz22, Pok18, Pok23}. Recent experimental studies indicate that the benefits of DD extend to dynamic circuits, where appropriately applied sequences can suppress errors during mid-circuit measurements~\cite{Bau23,Bau24,Shi24,Hot25,shirgure2026errormitigationdynamiccircuits}. However, these approaches typically rely on locally fixed DD sequences. As a result, such approaches lack the ability to tune DD sequences to qubit-dependent noise, circuit structure, or problem-specific constraints, despite growing evidence that tailored and learned DD protocols can yield substantial performance gains~\cite{Ton24,Sei24,Coo24,Rah24}. This limitation is particularly acute in dynamic circuits, where experiments across multiple platforms have shown that MCM implementations induce spatially structured errors on nearby qubits~\cite{Gae21, Mos23, Hot25}, with both the strength and character of these errors depending on the specific qubit being measured~\cite{Gov23,Shi24,Zho25}.

Here, we demonstrate an empirical protocol for learning multi-qubit dynamical decoupling (DD) sequences tailored to errors induced by mid-circuit measurements in dynamic circuits. By optimizing DD sequences directly on hardware, our approach is efficient and scalable to large circuit sizes without requiring full process tomography or exponential classical resources. Central to our approach is a spatiotemporal motif structure derived from the physical topology of mid-circuit measurements, which induces errors that are both qubit-specific and layer-dependent — a structure that generic, qubit-agnostic DD strategies do not capture. We show that the resulting learned sequences effectively suppress dynamic-circuit errors in randomized benchmarking experiments while remaining compatible with multi-qubit operation. We then apply the protocol to the quantum Fourier transform with measurement (QFT+M), where under learned DD control, we observe enhanced process fidelity on connected chains of up to 20 qubits. Counterfactual experiments in which learned sequences are applied without qubit specificity, or permuted among subregisters, confirm that this improvement stems from aligning optimization with spatiotemporal circuit structure. Finally, we apply QFT+M under learned DD control acting on a family of 10-qubit Greenberger–Horne–Zeilinger (GHZ) states and observe enhanced signal fidelity and stable many-body interference patterns in a regime that provides a stringent probe of correlated measurement-induced errors. Together, these results establish empirical DD learning as a practical and scalable strategy for suppressing measurement-induced errors in dynamic quantum circuits.

We introduce our DD learning protocol for dynamic circuits in Sec.~\ref{sec:dynamic_circuit_intro} and evaluate the effectiveness of the resulting sequences using mid-circuit measurement and dynamic circuit randomized benchmarking protocols~\cite{Gov23,Shi24} (Sec.~\ref{sec:application_and_benchmarking}). We demonstrate an application to QFT+M in Sec.~\ref{sec:dd_for_qftm} and high signal-to-noise QFT+M of GHZ states in Sec.~\ref{sec:qft_ghz_results}, and conclude in Sec.~\ref{sec:conclusions}.

\begin{figure*}[ht]
	\centering
	\includegraphics[width=\linewidth]{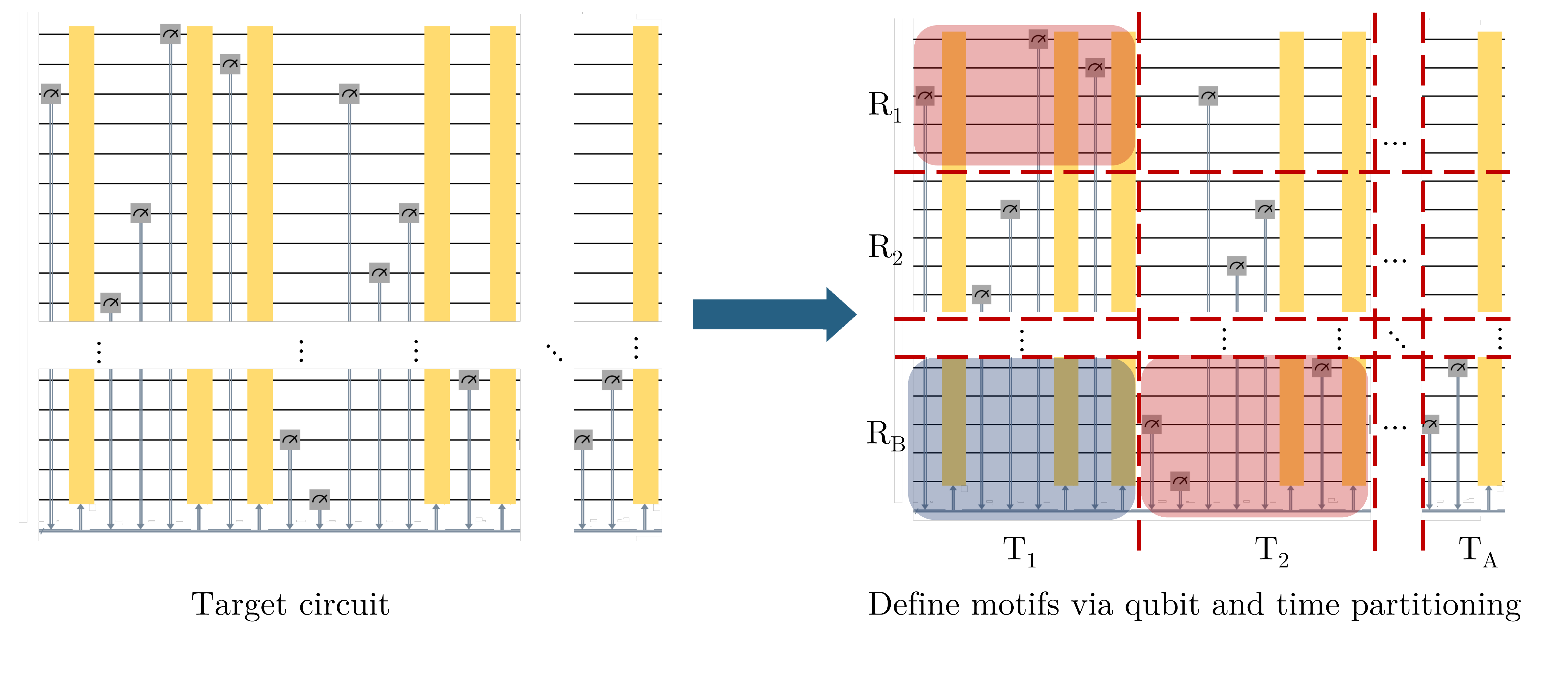}
	\caption{(left) A general dynamic circuit consists of repeated layers where some subset of ``measured'' qubits receive MCMs (gray) while the remaining ``unitary'' qubits remain idle. Conditional on the measurement result, subsequent operations (yellow) are then applied. These dynamic circuits allow for alternative constructions for quantum state evolution as realized in ``traditional'' quantum circuits, where purely unitary state evolution is followed by a final measurement. (right) To efficiently learn optimal DD sequences for the large circuit containing many distinct MCMs, we partition the circuit into $A$ temporal intervals and $B$ qubit registers and treat each qubit-time partition, such as $(T_1, R_1)$ and $(T_2, R_B)$ highlighted in red, for optimization with GADD. In contrast, motifs such as $(T_1, R_B)$ highlighted in blue have no MCMs and are excluded from training for efficiency.}
	\label{fig:dynamic_motif_generation}
\end{figure*}
\section{Dynamic circuits and empirical DD learning}
\label{sec:dynamic_circuit_intro}

Dynamic circuits are composed of layers of MCM and FF operations, as well as conditional unitary operations (Fig.~\ref{fig:dynamic_motif_generation}). In each MCM and FF layer, we categorize circuit qubits as either ``measured'' or ``unitary''; i.e., they either do or do not experience MCMs. Noise suppression in dynamic circuits is particularly important as dynamic circuit noise can lead to unintended conditional operations~\cite{Koh26}, which can quickly propagate due to the non-local interactions mediated by mid-circuit measurement and feedback~\cite{Bau23}. However, finding the exact DD sequence for optimal error suppression is difficult for multiple reasons. First, as DD sequences rely on balancing terms in the error Hamiltonian between forward and reverse time evolution, corresponding to commuting and anticommuting DD pulses respectively, one requires complete knowledge of the temporally- and spatially-correlated noise spectrum to design DD sequences from first principles. This challenge is exacerbated by sensitivity to environmental noise~\cite{Ton24, Sei24} and quantum control restrictions on real quantum devices~\cite{Byl11, Kra19, Zha22_2}, such as finite pulse widths, discrete device timings, and systematic gate errors. Even with an exact, ideal error model of the time- and qubit-correlations of all errors on a quantum device, it is combinatorially challenging to determine optimal DD sequences for each qubit on a many-qubit device~\cite{Zan99,Lid14, Cho20, Zhou23, Zhou24}. 

Beyond the aforementioned challenges, MCM-induced errors in dynamic circuits, which occur in all leading quantum computational platforms due to readout operations feeding back on adjacent device qubits~\cite{Gae21, Mos23, Gov23, Shi24, Hot25}, make \textit{a priori} determination of theoretically ideal DD sequences difficult. Recent benchmarking experiments~\cite{Gov23,Shi24,Hot25} indicate that individual device qubits induce distinct errors on neighboring qubits. Thus, as dynamic circuit execution proceeds, the qubit-correlated noise environment varies layer-by-layer~\cite{Sei24}. As a result, DD strategies must be tuned to suppress both spatial and temporal measurement-induced errors.
\begin{figure*}[ht]
	\centering
	\includegraphics[width=\linewidth]{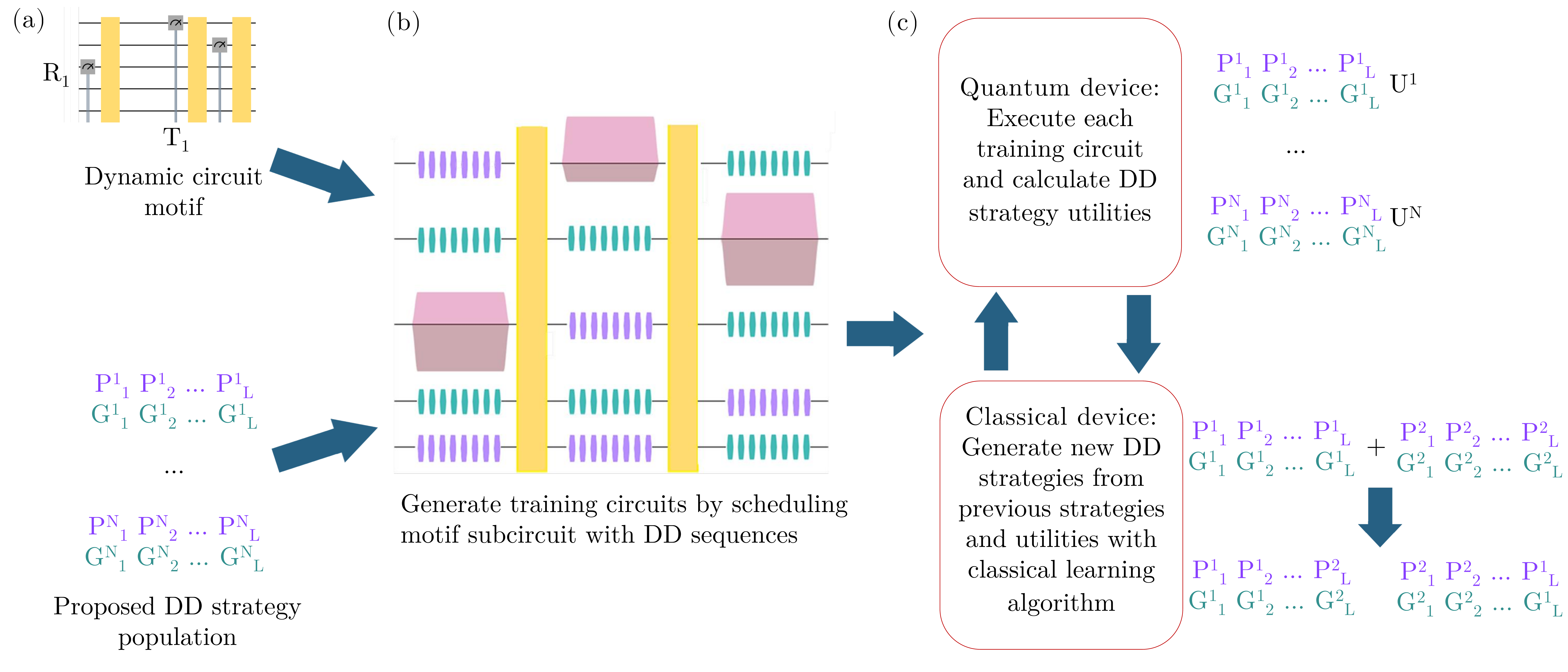}
	\caption{(a) Two distinct sequences (teal and purple) are independently optimized for the depicted DD motif. Although in principle as many colors can be used as desired, we show $k = 2$, which achieves basic crosstalk cancellation on 1D chains~\cite{Zho23}, while increasing $k$ increases the time to convergence. (b) The dynamic circuit is scheduled, with the long measurement instructions (pink) leaving idle periods on unitary qubits for DD sequence placements. Colors are assigned to idle periods according to odd and even parity distance respectively from the nearest measured qubit. (c) After all training circuits are generated and colored, the GADD method presented in Ref.~\cite{Ton24} is used to optimize DD through iteration between calculating utilities $\{U^n\}$ from quantum computational outcomes and determining new DD sequences through classical learning algorithms.}
	\label{fig: dynamic_empirical_learning}
\end{figure*}

Our protocol overcomes the aforementioned challenges by employing classical learning to empirically optimize DD sequences, treating sequence pulses as parameters of a utility function subject to extremization. We use a genetic algorithm~\cite{Mit99} for optimization, based upon previous computational work in DD sequence learning~\cite{Jud92,byrd2003empirical,Qui13,Ton24}. In particular, we define a utility function on empirical circuit performance data with a population of DD sequences on quantum devices and use a classical computer to analyze the results and generate new DD sequences; the specific utility function used in our experiments is described in Appendix~\ref{sec:genetic_algo_training_utility_appendix}. Through this use of quantum-classical feedback, the iterative method converges to DD strategies sensitive to both global and local qubit- and time-dependent coherent noise commensurate with their presence on the target device.

GADD operates on \textit{DD motifs}, subcircuit structures from the target circuit representative of the error environment~\cite{Ton24}, for efficient DD learning. Sequences learned for each subcircuit motif are then applied back to the target circuit, thus allowing sequences to be optimized to errors specific to local quantum operations or local timing and quantum control restrictions. Furthermore, convergence and computational time per iteration would be intractable without breaking sufficiently large circuits into the present motifs.

The spatially localized and temporally layered character of MCM-induced errors directly motivates a corresponding spatiotemporal partitioning of the circuit into DD motifs: by aligning qubit subregisters and time subintervals with the topology of individual measurement operations, learned sequences can be tailored to the locally distinct error environments each MCM induces (Fig.~\ref{fig:dynamic_motif_generation}). These motifs allow parallelized optimization of DD sequences for each subcircuit, addressing the spatiotemporal nature of measurement-induced errors in dynamic circuits. Furthermore, retaining the MCM and FF dynamic circuit structure within each subcircuit with the same logical control flow allows implicit DD sequence tuning to FF times. We note that our learning approach predates the \verb|stretch| functionality introduced in the utility scale dynamic circuit implementation \cite{IBMQuantumDynamicCircuits2025}, which allows for placeholder delays between operations such that relative pulse timings can be controlled without knowing the exact duration of the idle delays. While implemented using exact delays, our protocol can naturally incorporate the \verb|stretch| functionality.

With the motifs identified, we proceed to tuning DD sequences to each motif by applying the protocol shown in Fig.~\ref{fig: dynamic_empirical_learning}. The genetic algorithm generally converges as high-performing DD strategies are both (a) selected for after each iteration and (b) more likely to engage in reproduction~\cite{Qui13}. Upon convergence, the highest utility sequence for the given motif is taken and applied to the respective position within its motif in the original target dynamic circuit, as we demonstrate in experiments on the QFT+M dynamic circuit in Sec.~\ref{sec:dd_for_qftm}. In total, we apply the following steps for DD sequence learning:

\begin{enumerate}[(i)]
    \item Partition the total time into intervals $T_1, T_2, \dots, T_A$, allowing for better tuning of DD sequences to specific MCM-induced errors present in the motif instead of arriving at generic sequences with reasonable average-case performance.
    \item Partition circuit qubits into subregisters $R_1, R_2, \dots, R_B$, allowing the empirical learning algorithm to find distinct sequences that can suppress measurement-induced errors associated with individual measurement operations, while also being sufficiently large to capture various sources of crosstalk and qubit-correlated noise~\cite{Her21, Gov23, Ezz22, Bla21, Hey24}.
    \item Identify motifs $M_{ij} \equiv (T_i, R_j)$ involving MCMs for DD learning and group for GADD parallelization. Two motifs $M_{i_1j_1}$ and $M_{i_2j_2}$ can be parallelized if qubit correlations between the two motifs are trivial.
    \item For each group of motifs identified for parallelization, execute the learning algorithm depicted in Fig.~\ref{fig: dynamic_empirical_learning}.
    \item Pad the resulting learned DD sequences for each motif on the corresponding locations in the initial target circuit.
\end{enumerate}

Details for creating training circuits for empirical DD learning on a given dynamic circuit are discussed in Appendix~\ref{sec:genetic_algo_training_utility_appendix}. Specific discussion regarding the application of the empirical learning algorithm for DD optimization on limited-connectivity superconducting-qubit quantum computers and addressing device- and platform-specific errors is discussed in Appendix~\ref{sec:applying_dd_to_training_appendix}, and an example application is found in Appendix~\ref{sec:qftmtraining_results_appendix}. Runtime, resource allocation, and efficiency of the algorithm are discussed in Appendix~\ref{sec:scaling}, and detailed characterization and experimental data relating to the GADD method can be found in Ref.~\cite{Ton24}.

\section{Learning and benchmarking DD sequences for dynamic circuits}
\label{sec:application_and_benchmarking}
We perform empirical optimization experiments to learn DD sequences for dynamic circuits on \verb|ibm_kyiv|, one of IBM's Eagle superconducting qubit devices. To rigorously assess the effectiveness of the learned sequences, we benchmark performance using two complementary protocols: mid-circuit measurement randomized benchmarking (MCM-RB)~\cite{Gov23} and dynamic circuit randomized benchmarking (DC-RB)~\cite{Shi24}. MCM-RB isolates measurement-induced error channels in a controlled few-qubit setting aligned with the DD training configuration, allowing direct validation of the learned suppression mechanisms. In contrast, DC-RB evaluates performance in a larger-scale dynamic circuit executed on a 20-qubit chain, providing an application-independent and experimentally realistic test of learning performance. Together, these benchmarks establish both targeted suppression in the training setting and generalization of DD learning from motifs to scalable dynamic circuits.
\subsection{Learning on QFT+M and MCM-RB}
\label{sec:learningandmcmrb}
In a QFT+M circuit with $N$ qubits, $N$ circuit layers with mid-circuit measurement operations are concatenated. In the $m^{\text{th}}$ circuit layer, $0\leq m \leq N-1$, the $m^{\text{th}}$ qubit is measured, and a conditional $R_k \equiv \begin{pmatrix}
        1 & 0 \\
        0 & e^{2\pi i/2^k}
    \end{pmatrix}$ 
operation is applied to all subsequent qubits $m' > m$ for $k = m' - m + 1$ if qubit $m$ is measured to be in $\ket{+}$ in the desired basis.

In our learning step, we perform QFT+M circuit partitioning as described in Sec.~\ref{sec:dynamic_circuit_intro} into 5-qubit subregisters and their respective temporal intervals (Fig.~\ref{fig:dynamic_qft_mcmrb}a). For the 30-qubit chain depicted, we arrive at 6 motifs $M_{11}, M_{22}, \dots, M_{66}$, where $M_{ij}$ denotes the circuit motif of time interval $T_i$ and register $R_j$, with MCMs present. To increase DD sequence learning efficiency, we parallelize such that learning on $M_{ii}$, $M_{(i+3)(i+3)}$ occur simultaneously. We perform 10 GADD iterations for each of the resulting 3 parallelized blocks; details of training are described in Appendix~\ref{sec:genetic_algo_training_utility_appendix}.

\begin{figure*}
    \centering
    \includegraphics[width=\linewidth]{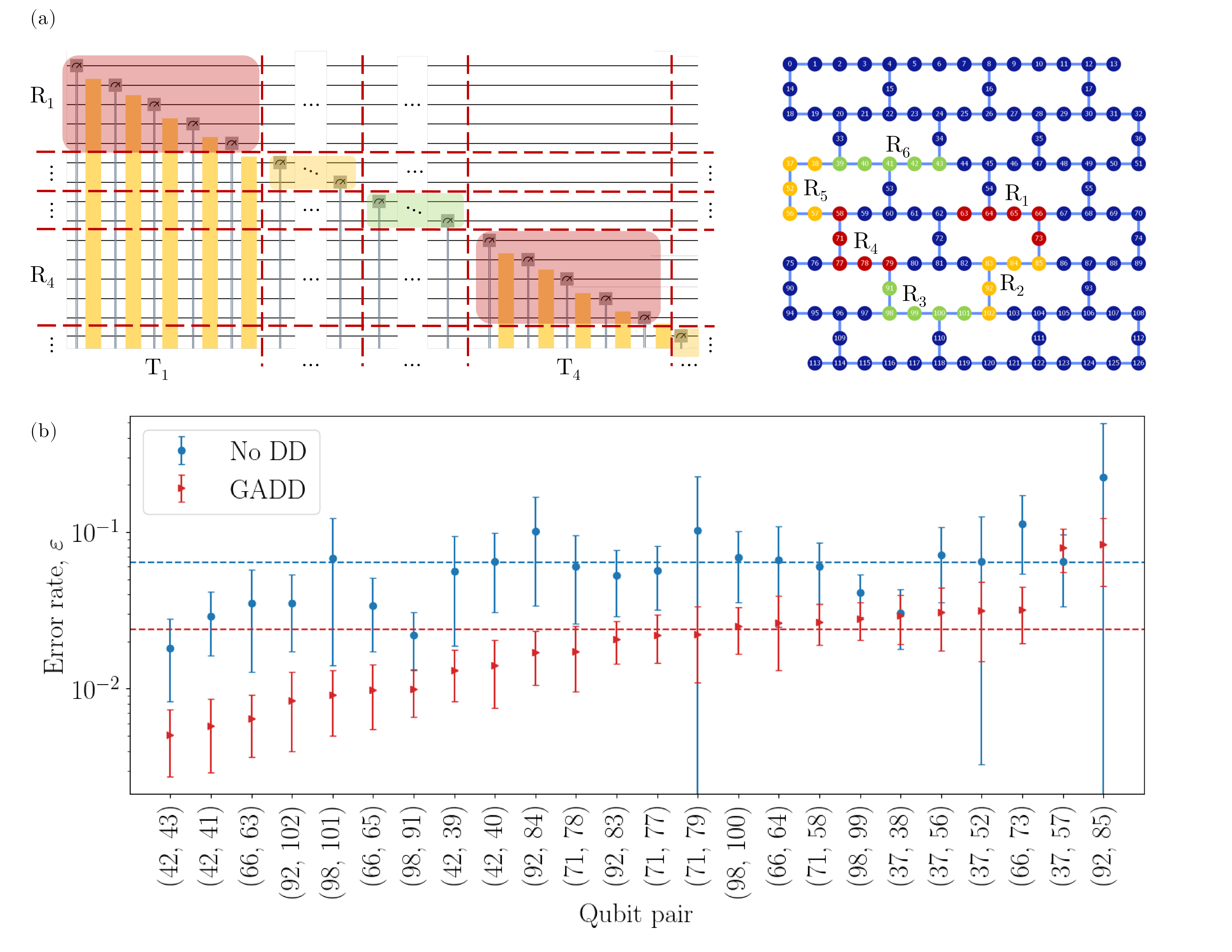}
    \caption{(a) We optimize DD strategies for the QFT+M circuit with an initial choice of 30 device qubits in a 1D chain. Chain qubits are associated with six partitions such that DD motifs on pairs of partitions marked with identical colors can be optimized in parallel. Focusing on $R_1$ and $R_4$, the associated training circuit is constructed by concatenating the QFT+M circuit on five qubits, which has an easily computable outcome distribution, on each partition. (b) Error per layer statistics of all sampled qubit pairs from the QFT+M circuit, sorted by post-GADD error rate (red), along with their corresponding pre-DD error rate (blue). Dashed lines represent error per layer values for the corresponding DD strategies.}
    \label{fig:dynamic_qft_mcmrb}
\end{figure*}

We assess the performance of these sequences learned from the QFT+M circuit motifs to suppress errors associated with mid-circuit measurement operations through the \verb|mcm-rb| protocol from the larger MCM-RB suite~\cite{Gov23}. In this protocol, circuits are constructed on two qubits, one unitary and measured respectively, where the unitary qubit receives a fixed-length $l$ sequence of random Clifford operations while the measured qubit receives an MCM interleaved in the circuit timeline between all pairs of adjacent measurements. Then, the probability that the unitary qubit, prepared in the $\ket{0}$ state, remains as $\ket{0}$, is measured as a function of $l$ via averaging over many random length $l$ circuits. The error rate per \verb|mcm-rb| layer (EPL) is then extracted through fitting $p(\text{unitary qubit }\ket{0}) = A\alpha^l+ B$, where the EPL $= (1-\alpha)/2$~\cite{Kni08, Gov23}.

We benchmark the performance of our empirically learned DD strategies on each circuit motif by fixing one measured qubit, then considering the respective qubit pairs with all other qubits in the motif, for $24$ total MCM-RB experiments. For each experiment, we calculate the EPL without DD; then, we implement the GADD sequence associated with the color of the idle period on the unitary qubit during measurement during all idle periods in the \verb|mcm-rb| circuit and calculate the EPL again. As the random Clifford operations perform a single-subsystem twirl on existing always-on unitary circuit errors, the difference in EPL with and without application of GADD strategies quantifies the learned sequences' ability to suppress errors associated with mid-circuit measurement.

We show the EPL associated with the unitary qubit in each tested qubit pair in Fig.~\ref{fig:dynamic_qft_mcmrb}b, representing comprehensive benchmarking of GADD sequences for each motif. All raw results of \verb|mcm-rb| experiments can be found in Appendix~\ref{sec:mcmrb_results_appendix}. As depicted, our empirically learned sequences significantly decrease the EPL, the key figure of merit, in all six DD motifs for which learning was performed, and in all but three qubit pairs (which suffer from measured-/unitary-qubit frequency collisions and thus have unreliable DD learning; see Appendix~\ref{sec:applying_dd_to_training_appendix} and~\ref{sec:mcmrb_results_appendix} for further discussion). Indeed, we reduce the median EPL of unitary qubits across all motifs from 0.060 without DD to 0.021 with our optimized DD sequences, as well as the mean unitary qubit EPL from 0.064(14) to 0.023(2). Such a three-fold improvement in benchmarked circuit error rates translates directly to improved dynamic circuit performance and demonstrates the ability to learn measurement-error suppressing DD strategies directly from subcircuit motifs from quantum applications.

\subsection{DC-RB}
\begin{figure*}[t]
    \includegraphics[width =\linewidth]{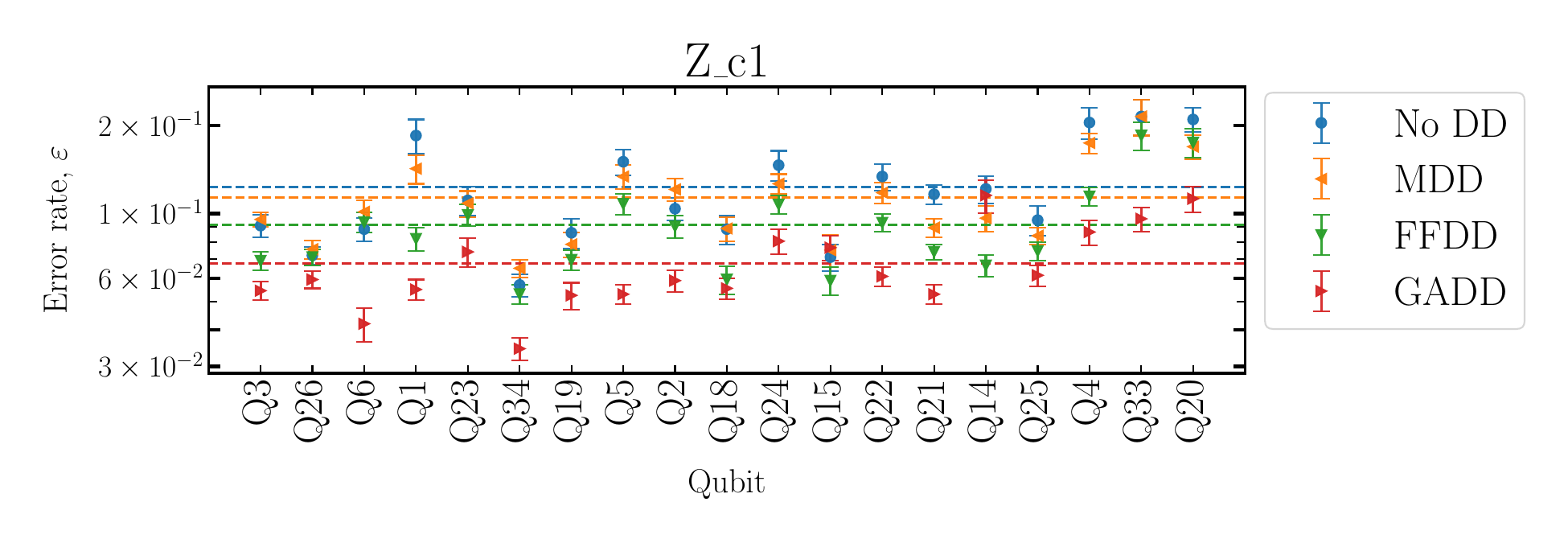}
    \includegraphics[width =\linewidth]{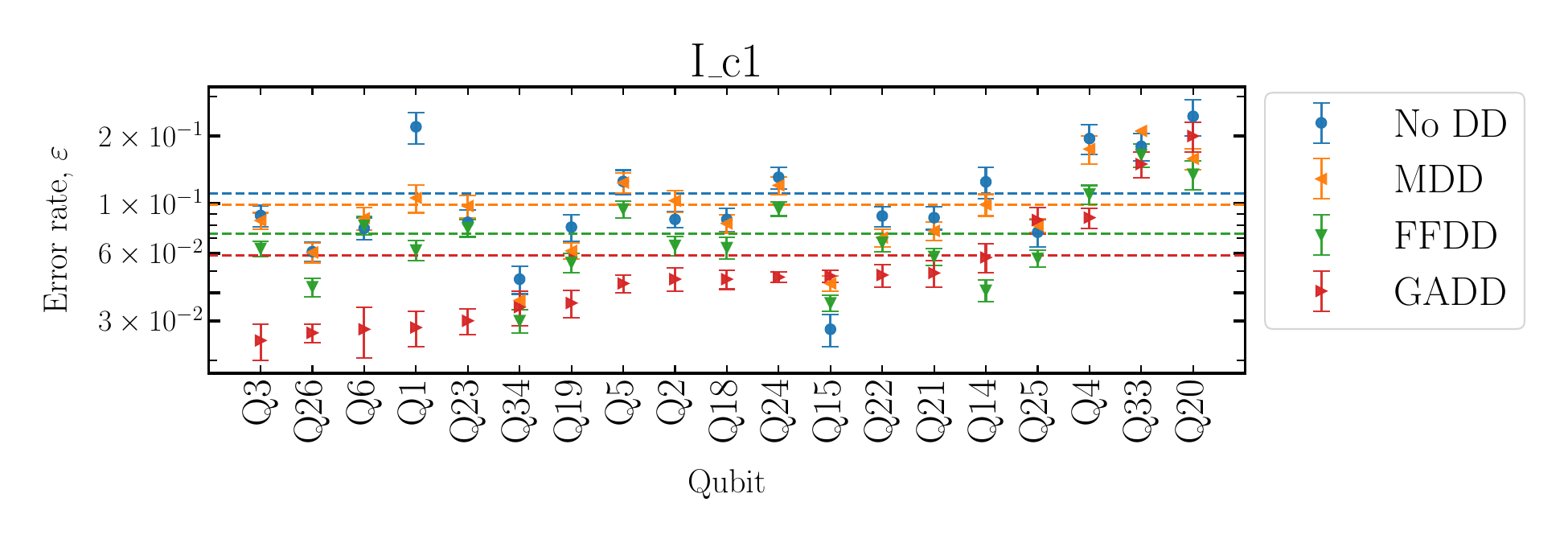}
    \cprotect\caption{EPL for each qubit on which the DC-RB blocks (top) \verb|Z_c1| and (bottom) \verb|I_c1| described in
  Ref.~\cite{Shi24} were applied in parallel on 20 qubits. GADD outperforms both the fixed-time MDD and FFDD sequences, which require prior estimation of the feedforward time $\tau_{\mathrm{FF}}$, across the register, capturing both always-on and measurement-induced crosstalk errors that theoretically derived sequences fail to suppress. Qubit labels on the x-axis denote the device qubit on \verb|ibm_kyiv| receiving dynamical decoupling operations during
  measurement of device qubit 0. Dashed lines represent error per layer values for the corresponding DD strategies.}
    \label{fig:dcrb_results}
\end{figure*}

We proceed beyond the QFT+M application setting by comparing empirically learned GADD sequences to existing, theoretically derived DD sequences on the dynamic circuits randomized benchmark (DC-RB). In contrast to MCM-RB, which characterizes measurement-induced errors in a controlled few-qubit setting, DC-RB provides an application-independent assessment of learned DD performance at 20-qubit scale. In DC-RB, we interleave two types of dynamic circuit blocks, \verb|Z_c1| and \verb|I_c1|, between one-qubit randomized benchmarking sequences of unitary qubits, on registers of fixed unitary and measured qubits~\cite{Shi24}. These circuit blocks initialize the measured qubit in $\ket{1}$ with an $X$-operation. Conditional on measurement of ``1'', a $Z$-gate or identity operation respectively is then applied to each unitary qubit and the measured qubit is reset by an $X$-gate. Note that in the case of the \verb|Z_c1| block, readout assignment errors are fed into unitary qubits as Pauli errors. We apply the GADD protocol in parallel on 20 connected qubits on \verb|ibm_kyiv|. In the register of connected qubits, we have both ``always-on'' traditional unitary crosstalk as well as measurement-induced qubit-correlated errors; through empirical optimization, we learn sequences that suppress both sources of error. 

We train GADD separately on the \verb|Z_c1| and \verb|I_c1| blocks using motif circuits consisting of a specific randomized Clifford sequence of length $l=3$, distinct from the QFT+M training circuits. The resulting sequences are applied to idle periods during MCM and FF in random DC-RB circuits of varying depth $l$, and $p(\text{unitary qubit in } \ket{0})$ is fit as a function of $l$ to extract the error per DC block layer. The effectiveness of sequences trained on short $l=3$ motifs across circuits up to $l=35$ directly demonstrates generalization of learned DD strategies across circuit depth.

We compare empirically learned GADD sequences against DD sequences theoretically derived for dynamic circuits: MDD, which applies two $X$ pulses during $\tau_{\mathrm{M}}$, and FFDD, which applies four $X$ pulses scheduled using prior knowledge of $\tau_{\mathrm{M}}$ and $\tau_{\mathrm{FF}}$~\cite{Bau24,Shi24}. As shown in Fig.~\ref{fig:dcrb_results}, GADD outperforms both purpose-built sequences across the 20-qubit register. For the \verb|Z_c1| block, the mean EPL with GADD is 0.068(5), compared to 0.091(8) with FFDD, 0.11(1) with MDD, and 0.12(1) without DD; analogous improvements are observed for \verb|I_c1|, with mean EPL of 0.059(10) versus 0.074(7), 0.099(10), and 0.11(1) respectively. Full results are included in Appendix~\ref{sec:dcrb_results_appendix}.

The advantage of GADD over MDD and FFDD, sequences explicitly optimized for dynamic circuits using calibrated feedforward timing, reveals that theoretically derived sequences fail to capture the full error environment of the 20-qubit chain. In particular, GADD captures both coherent ``always-on'' errors such as $ZZ$-crosstalk between neighboring qubits as well as measurement-induced crosstalk errors, neither of which MDD and FFDD suppress. GADD achieves this without requiring prior knowledge of $\tau_{\mathrm{FF}}$ or device-specific calibration inputs (Fig.~\ref{fig:dynamic_qft_mcmrb_appendix}c). The persistence of this improvement across the 20-qubit register establishes empirical optimization as systematically outperforming purpose-built theoretical approaches in this setting.

\section{Learned sequences for QFT+M}
\label{sec:dd_for_qftm}
Having reduced dynamic circuit error rates through empirical DD learning, we apply the learned sequences to QFT+M circuits, demonstrating that benchmarked error suppression translates directly to improved performance in a practical quantum algorithm. In Ref.~\cite{Bau24}, the authors implement QFT+M on IBM superconducting qubit hardware and demonstrate nontrivial QFT process fidelities on up to 40 qubits through semiclassical fan-out operation, which allows for $O(N)$ circuit depth as opposed to $O(N^2)$ depth of the traditional unitary QFT for an $N$-qubit state~\cite{Cor21}. However, the demonstration is performed on a largely disjoint selection of circuit qubits over the quantum device for crosstalk minimization, while practical quantum algorithms require QFT+M on registers where qubits are densely connected via entangling operations. In particular, QFT+M on connected qubit registers is essential for any application where the input state is prepared via multi-qubit unitary evolution.

Here, we apply empirically learned GADD sequences to enhance QFT+M performance on a 1D qubit chain, where all qubits are connected via entangling operations to their two nearest neighbors, on the 127-qubit superconducting qubit quantum device \texttt{ibm\_kyiv}, making this operationally relevant regime accessible. We quantify QFT performance enhancements through the process fidelity $\fp$, defined as the squared overlap between the ideal and experimentally prepared output states averaged over a sample of input computational basis states~\cite{Bau24}.  

\begin{figure*}
    \includegraphics[width = \linewidth]{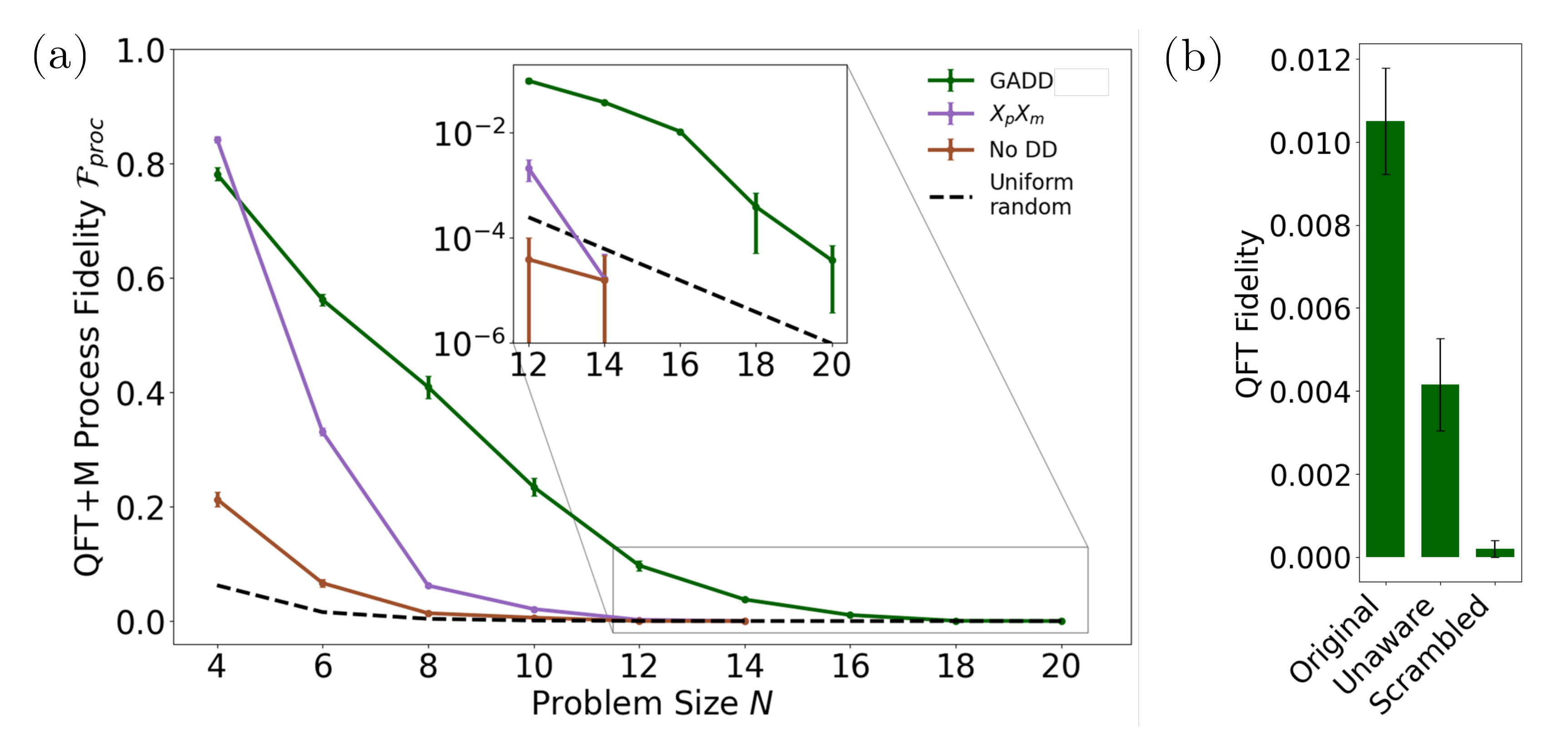}
    \caption{(a) Through implementation of empirical learning, we observe persistent nontrivial $\fp$ relative to canonical DD approaches, including nontrivial $\fp \geq 1\%$ in single experimental iterations for problem sizes $N = 8$--$16$. (inset) Relative to guessing computational basis states uniformly at random, with $\fp = 2^{-N}$ (dashed), empirically learned DD yields significantly greater $\fp$ through problem size $N = 20$, which is unachievable beyond $N = 14$ for the canonical $X_pX_m$ staggered sequence. (b) Under counterfactual applications of learned DD sequences at $N = 16$, where DD sequences learned on a single 5-qubit motif are applied to all 16 qubits (``Unaware''), and learned DD sequences are randomly allocated to motifs (``Scrambled''), $\fp$ drops significantly, demonstrating that empirical learning with motifs encoding information about local measurements is critical to successful DD optimization.}
    \label{fig:qft_scaling_results}
\end{figure*}

We select length-$N$ subchains for $N = 4,6,8,\dots,20$ from the $30$-qubit chain on which GADD sequences were learned. For each $N$, QFT process fidelities are calculated by sampling $16$ computational basis states $\ket{s}$ from a uniform distribution over all $2^N$ values of $s$. According to the protocol presented in Ref.~\cite{Bau24}, we calculate $\fp$ from efficiently preparing $\text{QFT}^\dag\ket{s}$ for each sampled $s$ via local single-qubit virtual state rotations and empirically estimate the amplitude of $\ket{s}$ in the state resulting from application of QFT+M with dynamical decoupling. The QFT+M circuit on $N$ qubits is constructed by appending semiclassical quantum Fourier transform circuit layers on selected qubits indexed by $m = 0, 1, \dots, N-1$, as described in Sec.~\ref{sec:learningandmcmrb}. 

Empirically learned dynamical decoupling sequences are added to all idle periods during MCMs in the QFT+M circuit. For a measured qubit $q_m$ belonging to register $R_i$, we apply the DD sequences learned on the motif $M_{ii}$ to all idle, unitary qubits during the measurement instruction. DD sequences are assigned to unitary qubits based on distance parity from the measured qubit (i.e., equivalently to qubit coloring for the corresponding motif). Furthermore, learned sequences are tested against performing the same experiment with no DD and with the naive application of the $X_pX_m$ staggered sequence over the same sample of $\{\ket{s}\}$ for each $N$. The $X_pX_m$ staggered sequence suppresses $Z$ phase and $ZZ$ crosstalk errors with DD in quantum processors~\cite{Zho23,Shi24dissipative} and has demonstrated relative success when evaluated against other canonical DD strategies in the unitary setting~\cite{Ezz22}, making it a common choice for applying DD to a circuit without \textit{a priori} knowledge of the backend noise, a device- and measurement error-agnostic DD strategy that can be applied easily by end-user quantum computation practitioners.

We compare QFT+M process fidelity $\fp$ results as a function of $N$ when GADD sequences are applied with the measurement error-agnostic $X_pX_m$ staggered strategy in Fig.~\ref{fig:qft_scaling_results}a. Results without error suppression due to DD are included for comparison, where $\fp < 1\%$ for $N \geq 8$, matching the behavior in Ref.~\cite{Bau24} with even more substantial fidelity decay due to the linearly connected qubit selection. Upon application of $X_pX_m$ strategy, we observe much-improved fidelities at small problem sizes, but reach the $\fp < 1\%$ threshold for $N \geq 10$, comparable to $\fp < 1\%$ threshold for $N \geq 11$ observed in the unitary QFT with DD as reported in Ref.~\cite{Bau24}. Although the improvement in $\fp$ relative to the results without DD is qualitatively consistent with previous successes in implementing theoretically derived DD sequences in the unitary setting to dynamic circuits~\cite{Bau24, Shi24}, the poor scaling with respect to $N$ is reflective of the challenge of suppressing MCM-induced errors that a naive method agnostic to measurement errors such as staggered $X_pX_m$ does not address.

By implementing DD sequences learned in the error context of the device, we access a regime of significantly greater $\fp$ across all $N$, extending the advantage provided by dynamic circuits to connected qubit subsets. For intermediate problem sizes of $8 \leq N \leq 14$, we achieve an over ten-fold increase in $\fp$ when compared to $X_pX_m$. Beyond $N = 14$, the dynamic QFT implemented with the $X_pX_m$ staggered strategy fails to yield signal above the $\fp = 2^{-N}$ threshold, representing the QFT fidelity associated with sampling computational basis states uniformly at random (Fig.~\ref{fig:qft_scaling_results}a inset). However, we exceed this threshold by at least a factor of $10$ up to the largest sampled problem size of $N= 20$, extending the useful range of the QFT+M subroutine. Furthermore, we achieve $\fp \geq 1\%$ up to $N = 16$ qubits, reflective of a $>$ 64-fold increase in the dimensionality of the Hilbert space on which QFT+M can be executed at the $\fp = 1\%$ level. At this problem size, we further evaluate the efficacy of empirical DD learning with counterfactual experiments (Fig.~\ref{fig:qft_scaling_results}b). We find that upon both assigning the learned DD sequences on the motif with the highest training utility across the entire 16-qubit circuit without qubit specificity, as well as randomly assigning learned DD sequences across motifs, $\fp$ drops significantly. By leveraging this local nature of noise correlation, we apply GADD to extend the range of problem sizes for which the QFT can be reliably applied on connected qubit chains. These results constitute direct experimental evidence that measurement-induced errors in dynamic circuits carry spatially localized structure that generic, qubit-agnostic DD strategies cannot capture — motivating the motif-based learning framework introduced here.

\section{Learned sequences for QFT+M of GHZ states}
\label{sec:qft_ghz_results}
\begin{figure*}
    \centering
    \includegraphics[width=\linewidth]{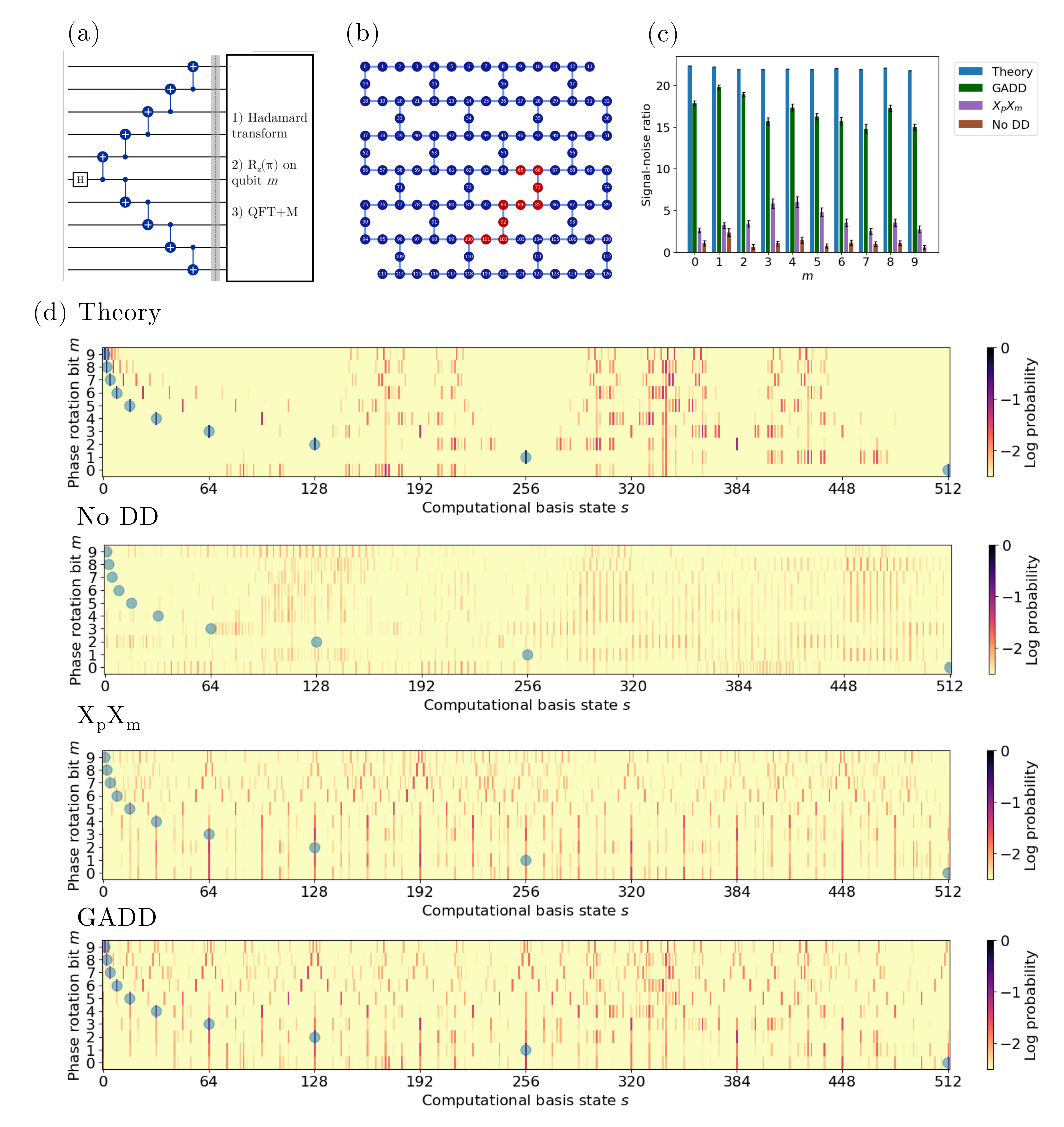}
    \cprotect\caption{(a) The family of states $\{\ket{\Psi_m}\}$ is constructed by a fan-out preparation of the Z-basis GHZ state, followed by a Hadamard transform to generate the X-basis GHZ state, then a $R_z(\pi)$ rotation on qubit $m$. (b) As the fan-out GHZ state preparation requires an adjacent set of device qubits, error suppression is necessary to address both always-on and measurement-induced correlated errors. We perform the QFT+M on the selected 10-qubit chain on the \verb|ibm_kyiv| device. (c) Signal-to-noise of the experimentally computed $\text{QFT}\ket{\Psi_m}$ via the semiclassical QFT with various error suppression strategies. Error bars are calculated by bootstrapping over 10000 experimental shots. (d) Visualization of $\text{QFT}\ket{\Psi_m}$ for all $m = 0, 1, \dots, 9$. The theoretically predicted peak at $s = 2^{N - m - 1} = 2^{9-m}$ (highlighted with large blue dots) is only observed with high signal-to-noise through empirically learned DD.}
    \label{fig:qft_ghz_main}
\end{figure*}

Partial measurement, feedback, and error detection are particularly important for studying many-body entangled states~\cite{Ris13,Zha17,Zha18,Jav25,Lia25}. In particular, GHZ states are maximally entangled many-body states whose collective phase sensitivity scales linearly with system size, enabling Heisenberg-limited precision in quantum metrology~\cite{Bol96,Gio04,Hor09}. Because their coherence is encoded in a single global off-diagonal term, GHZ states are exceptionally sensitive to spatially correlated noise across the register. They therefore provide a stringent test of measurement-resilient control strategies in dynamic quantum circuits. Furthermore, unitary fan-out preparation of GHZ states generates extensive entanglement on a 1D qubit chain (Fig.~\ref{fig:qft_ghz_main}a), implying that measurement-induced errors cannot be mitigated by choosing disconnected qubit subgraphs. In this regime, both always-on interactions and measurement-induced correlated errors act across the register, making optimal suppression nontrivial \textit{a priori}. 

With the nontrivial $\fp$ enabled by empirically learned DD in intermediate one-dimensional chains of length $N$, we apply QFT+M to entangled states prepared via unitary fan-out. Specifically, we consider the family
\begin{multline}
    \ket{\Psi_m} = \frac{1}{\sqrt{2}}\bigg( \ket{+}^{\otimes(N-m-1)}\ket{-}\ket{+}^{\otimes m}
\\ + \ket{-}^{\otimes(N-m-1)}\ket{+}\ket{-}^{\otimes m}\bigg)
\end{multline}
corresponding to an X-basis GHZ state on $N = 10$ qubits with a single $R_z(\pi)$ phase flip applied to qubit $m = 0, \dots, 9$. The Fourier transform of $\ket{\Psi_m}$ is governed by coherent interference among all $2^{N-1}$ computational basis components of the GHZ superposition, producing sharply localized many-body interference peaks whose position depends on the location of the phase flip. The constructive interference responsible for these peaks arises from the globally encoded phase structure of the state and is therefore highly sensitive to correlated phase errors across the register. In particular, $\text{QFT}\ket{\Psi_m}$ exhibits sharply peaked amplitude on computational basis states $\ket{2^{N-1-m}}$:
\begin{equation}
    |\bra{2^{N-1-m}}\text{QFT}\ket{\Psi_m}|^2 > \frac{2}{\pi^2} \approx 0.203,
\end{equation}
with proof given in Appendix~\ref{sec:qft_on_ghz_full}. As a result, the index $m$ is encoded logarithmically in the location of the dominant Fourier peak; however, accurate recovery of $m$ from $\ket{\Psi_m}$ for all $m$ requires suppression of measurement-induced errors on all $N$ qubits involved in the dynamic circuit.

We prepare $\ket{\Psi_m}$ on a 10-qubit chain of \verb|ibm_kyiv| (Fig.~\ref{fig:qft_ghz_main}b) and implement the semiclassical QFT. As the GHZ state is prepared identically across all error suppression conditions, the relative SNR improvements are attributable to DD performance during the QFT+M subroutine rather than differences in input state quality. In this system size, we observed nontrivial QFT+M $\fp$ under learned DD via random sampling of computational basis states, as shown in Sec.~\ref{sec:dd_for_qftm}. We quantify the behavior of these sequences on the QFT+M of $\ket{\Psi_m}$ via a signal-to-noise ratio (SNR) metric, defined as 
\begin{equation}
    \text{SNR} = \frac{P_{2^{N-1-m}} + P_{2^N - 2^{N-1-m}}}{2 \sqrt{\mathbb{E}_{s}\left[(P_s - \mathbb{E}[P_s])^2\right]}},
\end{equation}
where $P_s$ denotes the measured population of state $s$.

With the same learned DD strategies, we observe SNR values of $15$–$20$, approaching the finite-size theoretical bound $\text{SNR} < 22$ (Fig.~\ref{fig:qft_ghz_main}c) obtained by evaluating the ideal noiseless $N = 10$ QFT output distribution. Visual inspection confirms strong contrast at the predicted peak locations $s = 2^{N-1-m}$ for all $m$, recovering the expected logarithmic dependence (Fig.~\ref{fig:qft_ghz_main}d). In contrast, the measurement-agnostic $X_pX_m$ sequence yields substantially noisier results, and the absence of DD leads to further degradation. Quantitatively, empirically learned DD enhances SNR by a factor of three to four across all $m$ relative to $X_pX_m$. Achieving comparable signal levels with $X_pX_m$ would therefore require approximately $9$–$16$ times more experimental shots.

Beyond improved shot efficiency, the recovery of sharply peaked and symmetry-constrained Fourier structure demonstrates preservation of genuinely many-body interference in a maximally entangled state under dynamic circuit execution. The stability of these interference patterns indicates the ability of empirically learned DD strategies to suppress spatially and temporally correlated measurement-induced noise that would otherwise rapidly degrade many-body coherence. 

\section{Conclusions}
\label{sec:conclusions}
Effective and efficient error suppression is essential for harnessing the full capabilities of dynamic quantum circuits. At the same time, identifying optimal dynamical decoupling sequences in this setting is challenging due to the time- and qubit-dependent nature of measurement-induced errors. In this work, we introduce an empirical learning protocol that optimizes DD sequences as a function of both temporal subinterval and qubit subregister, enabling layer-dependent error suppression tailored to dynamic circuit execution. Using mid-circuit measurement and dynamic circuit randomized benchmarking, we observe two- to three-fold reductions in average dynamic circuit error rates from empirically learned DD sequences relative to undecoupled circuits. We further apply our approach to the dynamic circuit implementation of the quantum Fourier transform, extending previous demonstrations to connected qubit chains and achieving substantial improvements in QFT+M process fidelity. Leveraging the resulting signal enhancements, we realize high signal-to-noise QFT-based characterization of an entangled 10-qubit state in an experimental setting.

While the present demonstrations are performed on superconducting qubit hardware, the empirical learning framework is platform-agnostic and relies only on the ability to execute parameterized circuits and collect outcome statistics; extending these results to trapped-ion or neutral-atom platforms where dynamic circuits are increasingly relevant~\cite{Gae21,Fos23,Dec23,Gra23} is a natural direction for future work. As quantum hardware continues to scale and real-time classical control matures, dynamic circuits are expected to play an increasingly central role in quantum algorithms, error mitigation, and error correction~\cite{Tem17,Cai23}. Beyond the circuits explored here, measurement-conditioned dynamics enable the preparation and control of quantum many-body states with tunable entanglement structure, as demonstrated in recent experiments~\cite{Ski19,Fis23,Bau23,Che25}. More broadly, as quantum processors approach fault-tolerant operating regimes, suppressing measurement-induced errors during repeated syndrome extraction and adaptive control cycles will become increasingly critical~\cite{Ter15,Cra16,Goo25}. The empirical DD optimization framework presented here suppresses such errors without requiring detailed noise models, providing a practical pathway toward higher-fidelity quantum computation with dynamic circuits.

\section*{Data Availability}
The code implemented to generate the data that support
the findings of this article is openly available~\cite{gaddpackage}. The
raw data are not publicly available but are available upon reasonable request from the authors.

\section*{Author Contributions}
C.T. developed the empirical learning protocol for dynamic circuits, performed the experiments, and analyzed the data. L.S. contributed to the DC-RB protocol and performed the benchmarking experiments. E.H.C. and D.S.W. contributed to the QFT+M experimental design. B.P. conceived and supervised the project. C.T. wrote the manuscript with input from all authors.

\section*{Acknowledgements}
We thank Elisa Bäumer, Luke C. G. Govia, Swarnadeep Majumder, Maika Takita, Matthew Ware, Eric Zhang, and Helena Zhang for helpful discussions. We acknowledge support from the Army Research Office under QCISS (W911NF-21-1-0002). The views and conclusions contained in this document are those of the authors and should not be interpreted as representing the official policies, either expressed or implied, of the Army Research Office or the U.S. Government. The U.S. Government is authorized to reproduce and distribute reprints for Government purposes notwithstanding any copyright notation herein.

\bibliography{main}
\pagebreak
\appendix
\onecolumngrid

\section{Empirical learning details}
As defined in Ref.~\cite{Ton24}, empirical learning of DD requires defining how to create training circuits from a larger quantum circuit intended for application whose performance is well quantified by a selected utility function and how to apply test dynamical decoupling sequences to training circuits. Here, we address the specifics of the respective points in our application to the QFT+M dynamic circuit on superconducting qubit hardware and report our training results.

\begin{figure*}[ht]
    \centering
    \includegraphics[width=0.75\linewidth]{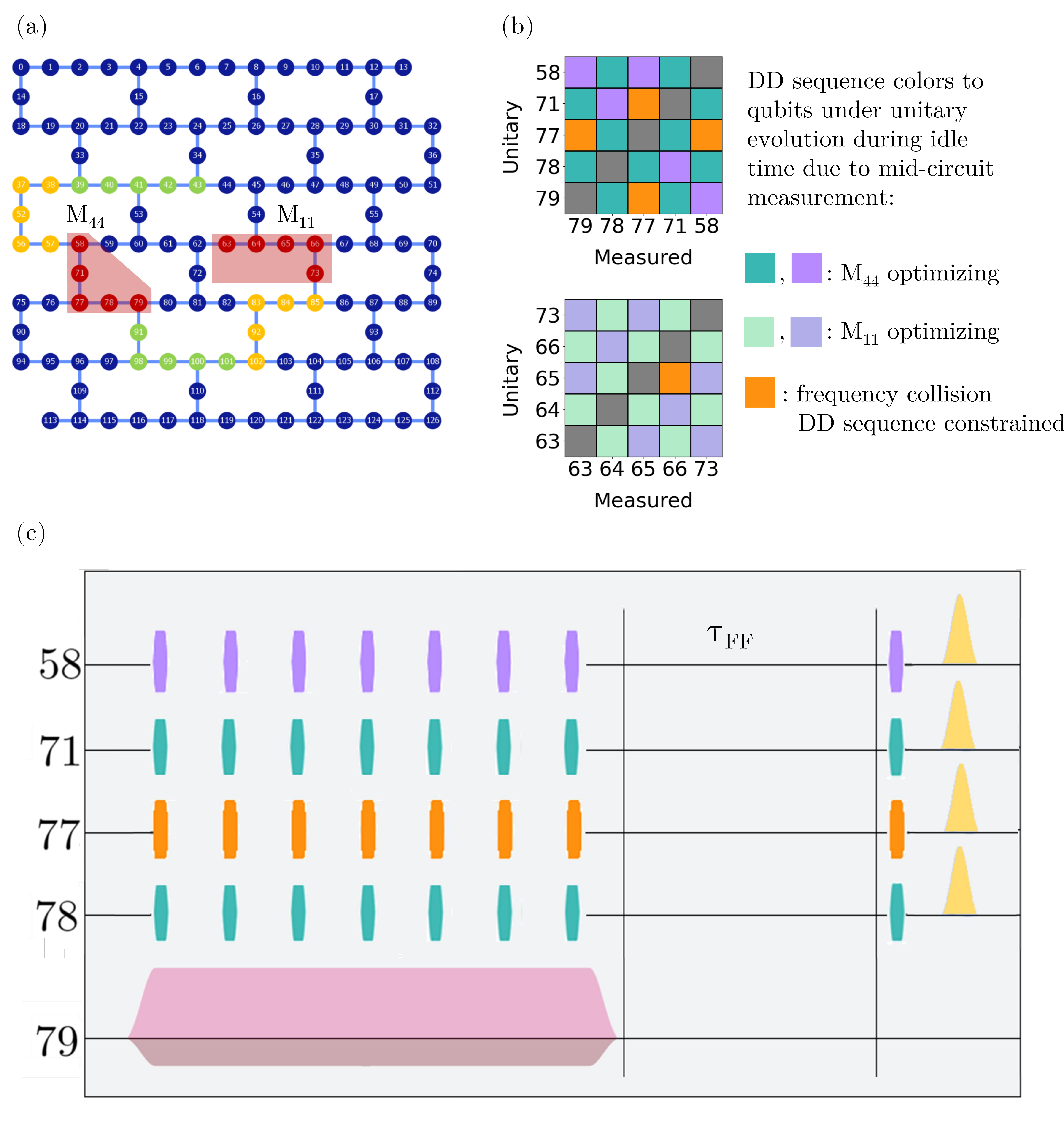}
    \cprotect\caption{(a) The 30-qubit 1D qubit chain shown in Fig.~\ref{fig:dynamic_qft_mcmrb} as a subset of the full \verb|ibm_kyiv| device. Qubit partitions are colored to denote motifs on which DD optimization is performed in parallel. (b) Using Eqn.~\ref{eqn:collision_condition}, we consider each measured-unitary qubit pair in the dynamic circuit. For the sake of demonstration, we mark in orange the qubit pairs $(m,u) = (77,79), (58,77), (79,77), (77,71), (66,65)$ in $R_4$ (top) and $R_1$ (bottom), as pairs for which frequency collisions were detected. The remaining measured-unitary pairs are colored according to the scheme described in Fig.~\ref{fig: dynamic_empirical_learning}. (c) The timeline of a sample subcircuit layer during DD optimization.  The collisional sequence in orange is the discussed constrained sequence of $I$ as its first six pulses and $X$ as its final two, to provide error suppression during  $\tau_{\mathrm{FF}}$ when the measurement pulse is off. In the case where constrained sequences occur on two qubits connected with an existing entangling operation on the device, the sequence on the second qubit is instead constrained to identity operations only as a heuristic approach to crosstalk cancellation~\cite{Zho23}.}
    \label{fig:dynamic_qft_mcmrb_appendix}
\end{figure*}

\subsection{Training circuit and utility function selection}
\label{sec:genetic_algo_training_utility_appendix}
We perform DD learning for the QFT+M circuit on the 30 qubits shown in Fig.~\ref{fig:dynamic_qft_mcmrb_appendix}a. We define DD motifs through partitioning all circuit qubits into six registers of five qubits as depicted in Fig.~\ref{fig:dynamic_qft_mcmrb}. Then, the relevant QFT+M circuit motifs correspond to the six subcircuits that lie on the ``diagonal'' of the QFT+M subcircuits; i.e., the subcircuits with nontrivial dynamic circuit behavior correspond to time intervals in sequential order equal to the order of the qubit subregisters. To quantify the behavior of these subcircuit structures in each iteration, we compare the empirically observed outcome distribution $\hat{\vb{p}}(s)$ of the training circuit with the theoretical distribution $\vb{p}$. For ease of computation, we change values of local $Z$-phase-gate parameters such that the training circuit for each motif corresponds to a 5-qubit QFT+M operation on the $\text{QFT}^{\dag}\ket{00000}$ state, as for any computational basis state $\ket{s}$, $\text{QFT}^{\dag}\ket{s}$ differs from $\text{QFT}^{\dag}\ket{0\dots 0}$ in only a qubit-dependent phase at each qubit. This preserves all circuit structure and order of MCMs, only changing circuit parameters. Thus, as previously demonstrated and discussed in the Cliffordization sections in Ref.~\cite{Ton24}, we expect that utility functions based on this circuit structure and $\vb{p}$ will be effective in an application setting.

We use the 1-norm similarity metric between $\hat{\vb{p}}$ and the theoretical distribution $\vb{p}$

\begin{equation}
    f_{\text{1-norm}} = 1 - \frac{1}{2} \sum_{s=0}^{2^N-1} \abs{\vb{p}(s) - \hat{\vb{p}}(s)}
    \label{eqn:1norm_utility}
\end{equation}
as the utility function for generating new DD populations in iterations of the genetic algorithm, where $N$ is the qubit number in the motif training subcircuit. With all hyperparameters and training circuits defined, we apply the following algorithmic steps to empirically learn DD sequences:
\begin{enumerate}[(i)]
  \item Determine genetic algorithm hyperparameters: DD sequence length $L$, number of independent DD sequences to optimize per motif (``colors'') $k$, number of distinct sets of $k$ sequences (``strategies'') to evaluate per iteration $K$, decoupling group $\mathcal{G}$ from which DD pulses are drawn
  \item Generate starting population of DD strategies (Fig.~\ref{fig: dynamic_empirical_learning}a). Genetic algorithm performance is maximized when, for each color and fixed pulse index $1 \leq i \leq L$, the elements of $\mathcal{G}$ are equally distributed across the $K$ strategies, i.e., each $g \in \mathcal{G}$ appears roughly $K/\abs{\mathcal{G}}$ times.
  \item Assign idle periods on unitary qubits in each dynamic circuit layer 
  to a color (Fig.~\ref{fig: dynamic_empirical_learning}b). We do this according to the distance of each unitary qubit on the quantum device to the nearest measured qubit: for shortest distance $d$ associated with a unitary qubit, the corresponding idle period receives color $d \pmod{k}$ for maximal sequence staggering~\cite{Niu24}.
  \item Schedule $K$ circuits; each circuit corresponds to implementation of its respective DD strategy, where idle periods assigned to any given color receive the respectively colored DD sequence.
  \item Execute the $K$ circuits on the quantum device and calculate the utilities $\{U^n\}$ associated with each of the outcome distributions over computational basis states for each circuit $1 \leq n \leq K$ (Fig.~\ref{fig: dynamic_empirical_learning}c). 
  \item Generate new candidate DD sequences via the genetic algorithm:
  \begin{itemize}
      \item Randomly select $K$ pairs of DD strategies with replacement from the population of $K$ sequences, with higher utility pairs more likely to be chosen
      \item Perform reproduction to arrive at $2K$ child DD strategies
  \end{itemize}
  \item Continue to iterate with the population of $K$ parent and $2K$ child DD strategies. After each evaluation of $3K$ DD sequences on the quantum device, only the $K$ highest-utility sequences are preserved before reproduction occurs. 
\end{enumerate}
In settings where GADD is performed in parallel on many motifs, a training circuit is created independently for each motif; then, the overall training circuit is created by aligning the start time for each circuit motif, which is possible as parallelizable circuit motifs must contain disjoint qubit registers. The utility function for each motif is then calculated by marginalizing away contributions of measured qubits belonging to other training circuits to the observed distribution $\hat{\vb{p}}$. In our experiment, we parallelize motifs pairwise according to the coloring in Fig.~\ref{fig:dynamic_qft_mcmrb_appendix}a and perform three separate learning experiments, corresponding to the colors of red ($M_{11}$ and $M_{44}$), yellow ($M_{22}$ and $M_{55}$), and green ($M_{33}$ and $M_{66}$) respectively. For convenience, we describe application of GADD to training circuits for the red motifs; the other two motif groups behave analogously.

\subsection{Applying DD strategies to training circuits}
\label{sec:applying_dd_to_training_appendix}
In Ref.~\cite{Ton24}, idle period ``coloring'' is implemented to study DD strategies, where composite DD sequences are represented by individuals in the genetic algorithm population, such that distinctly colored idle periods receive their respective distinct DD sequences from the strategy for improved suppression of correlated multi-qubit errors. Here, we extend DD coloring to dynamic circuits. Idle periods in the training circuit containing $m$ motifs from the target dynamic circuit are colored following the ansatz that measurement-induced errors behave similarly among qubits belonging to the same partition. As a result, each qubit partition within the training circuit is associated with its own set of colors. Here, we use two colors to denote the two independent DD sequences making up the strategy to be optimized on each circuit motif, corresponding to the shades of green and purple shown in Fig.~\ref{fig:dynamic_qft_mcmrb_appendix}b.

We assign colors to idle periods in order to suppress MCM errors according to established qualitative behavior of errors induced by MCMs on nearby idle qubits. Ref.~\cite{Gov23} establishes measurement-induced control error and measurement-induced two-qubit error as distinct qualitative error categories. Measurement-induced control errors occur when upon measurement of a given qubit in superconducting circuit QED systems, a resonator associated with the measured qubit is populated with a strong field~\cite{Bla21}. Photons associated with the field of the measured qubit resonator can occupy the resonator of adjacent unitary qubits, leading to unwanted coherent Z-phase errors on those qubits in the duration of the measurement pulse. Such measurement-induced Z-phase errors are amenable to suppression with DD~\cite{Shi24}. However, due to judicious hardware design and optimization, distinct resonators on a circuit QED device operate at different frequencies and more generally lie in distinct error environments on a superconducting quantum device, leading to distinct optimal DD sequences for each qubit. As the QFT+M circuit on a sequence of device qubits is constructed by measuring each qubit once, followed by a conditional $R_z$ operation on subsequent qubits, motifs can be defined simply according to the range of such measurement-induced control error, which we associate with 5-qubit subgraphs of the \texttt{ibm\_kyiv} device (Fig.~\ref{fig:dynamic_qft_mcmrb_appendix}a), and their associated intervals in time.

Measurement-induced two-qubit errors arise as the AC Stark shift associated with the measured qubit resonator field can bring the frequency of the measured qubit into near-resonance with neighboring unitary qubits. For a pair of qubits $q_1$, $q_2$ with coupling $J$ and operating frequencies $\omega$ and $\omega + \Delta$ respectively, the effective two-qubit Hamiltonian in a frame rotating at $\omega$ is~\cite{Gov23, Hey24}
\begin{equation}
    H = \frac{\Delta}{2}\sigma_{z,2} + J(\sigma_{+,1}\sigma_{-,2} + \sigma_{+,2}\sigma_{-,1}).
    \label{eqn:collision_hamiltonian}
\end{equation}
Although qubit frequencies in circuit QED systems are judiciously selected to ensure $\Delta$ is sufficiently large for qubits with nontrivial coupling $J$~\cite{Her21}, AC Stark shifts due to strong resonator fields can induce $J \sim \Delta$ between the measured qubit and a neighboring unitary qubit by decreasing $\Delta$~\cite{Dum24}. In this limit, the second term contributes nontrivially to the time evolution of the unitary qubit in the duration of the measurement instruction, leading to unwanted coherent excitation exchange.

Placing sequences tailored to suppress Z-phase errors in likely collisional settings will not benefit circuit performance. Thus, it is important to isolate potential instances of measurement-induced collision errors from the training circuits on which the DD strategy for coherent Z-phase errors is empirically optimized. Measurement-induced collisional errors only occur for small values of $\Delta$ between nearby qubits. To evaluate the potential for collisional error, we consider the transition frequencies $\omega_{01}$ and $\omega_{12}$ of all circuit qubits. In each circuit layer, we consider the difference frequencies
\begin{equation}
\begin{gathered}
    \Delta_1 = \abs{(\omega^{(m)}_{01} + \delta_s) - \omega^{(u)}_{01}} \\
    \Delta_3 = \abs{(\omega^{(m)}_{01} + \delta_s) - \omega^{(u)}_{12}}
    \label{eqn:collision_condition}
\end{gathered}
\end{equation}
over all pairs of measured qubits $m$ and unitary qubits $u$ within distance $4$ on the backend graph. Taking $\delta_s = -25$ MHz as a characteristic value for the measurement-induced Stark shift~\cite{Zha22}, we identify instances of $\Delta_1 \leq 17$ MHz and $\Delta_3 \leq 30$ MHz as possible collision instances. We neglect all other collision types as the measurement instruction pulse frequency is far detuned from all qubit frequencies, unlike the echoed cross-resonance pulse; thus, only the above (type 1 and type 3) collisions as classified in Refs.~\cite{Her21,Zhang22} are relevant. Furthermore, we consider pairs of $(m, u)$ beyond nearest neighbors due to multiplexing allowing for longer-range measurement-induced frequency collisions~\cite{Hei18}. In these settings, we apply a ``constrained'' sequence of $I$ as its first six pulses and $X$ as its final two, which is essential to avoid skewing the DD sequence learning process on the remaining qubit pairs. We call such sequences constrained as their placement and identity does not change during the GADD iteration; thus for clarity, we do not discuss these sequences as contributing to the behavior of any individual DD strategy in the genetic algorithm population. In our experiments, we perform such frequency tests among all qubit subregisters and apply collision-constrained sequences accordingly.

With determination of training circuits and allocation of idle periods in those training circuits to colors established, it remains to establish the timings of individual DD pulses within any given sequence. As established in Ref.~\cite{Ton24}, a significant advantage of the GADD method is its ability to sample all allocations of identity operations, which trivially belong to any decoupling group $\mathcal{G}$, among pulses in a DD sequence. This allows for many different DD sequence timing staggerings, corresponding to different suppression levels of different many-body error sources~\cite{Zho23}, to be sampled, without consideration of continuous-variable optimization for pulse time periods. As a result, we implement uniform pulse spacings during idle periods to allow the genetic algorithm to unbiasedly explore all allowed sequence staggerings. However, current device design prevents application of DD pulses simultaneous to classical measurement processing during the $\tau_{\mathrm{FF}}$. As a result, we generalize the FF-compensated approach introduced in~\cite{Bau23}, which established error suppression during feedforward by applying DD pulses immediately before and after the dedicated time interval by implementing the final pulse in our DD sequence to occur as the first instruction in all cases of the classical logic control flow such that the pulse will always occur immediately subsequent to $\tau_{\mathrm{FF}}$ (Fig.~\ref{fig:dynamic_qft_mcmrb_appendix}c). The remaining $L-1$ pulses are scheduled to be uniformly spaced during the MCM prior to FF. In our experiments, we use $L = 8$ DD pulses per sequence, which we denote in the format $P_1P_2P_3P_4P_5P_6P_7P_8$ for each $P_i \in \mathcal{G}$; thus, $P_8$ occurs after $\tau_{\mathrm{FF}}$ while $P_1$--$P_7$ are uniformly spaced before $\tau_{\mathrm{FF}}$. For example, the collision-constrained sequence is written as $IIIIIIX_pX_p$, allowing for reversed evolution of $Z$-phase errors during $\tau_{\mathrm{FF}}$, after the termination of the readout instruction and the frequency collision associated with the presence of the resonator field.

\subsection{QFT+M training and results}
\label{sec:qftmtraining_results_appendix}
For each motif, we optimize DD strategies containing sequences of length $L = 8$ and $k = 2$ colors, $K = 16$ strategies are preserved per iteration, and pulses are drawn from the decoupling group $\mathcal{G} =  \{I_p, I_m, X_p, X_m, Y_p, Y_m, Z_p, Z_m\}$. Here, $I_p = I_m$ is the identity element (notated as such for consistency) and $X_{p}, X_m$ and $Y_p, Y_m$ pulses are implemented through symmetrized unitary transformations with respective inverse directions of state rotation as described in Ref.~\cite{Vez24}. We note that $\mathcal{G}$ is not rigorously a group but behaves as such up to unimportant global phase differences in powers of $i$. We perform 10 GADD iterations per motif; the utilities (Eqn.~\ref{eqn:1norm_utility}) of each individual in the GADD population over all iterations are plotted in Fig.~\ref{fig:qftm_training_results}. Utility-maximizing DD sequences are listed in Table~\ref{table:topgaddsequences}. 

Once learning iterations are completed, the resulting DD sequences are padded to their corresponding motifs that were previously identified in the target dynamic circuit (Fig.~\ref{fig:dynamic_qft_mcmrb}a) in a quantum computational experiment. For experiments described in Sec.~\ref{sec:dd_for_qftm}, we select subchains of length up to $N=20$ from the initial chain of length $30$ and then apply DD strategies by considering which motif each qubit belonged to, then applying the strategy learned for the corresponding motif for idle periods on each respective qubit.

\begin{figure*}[t]
    \centering
    \includegraphics[width=\linewidth]{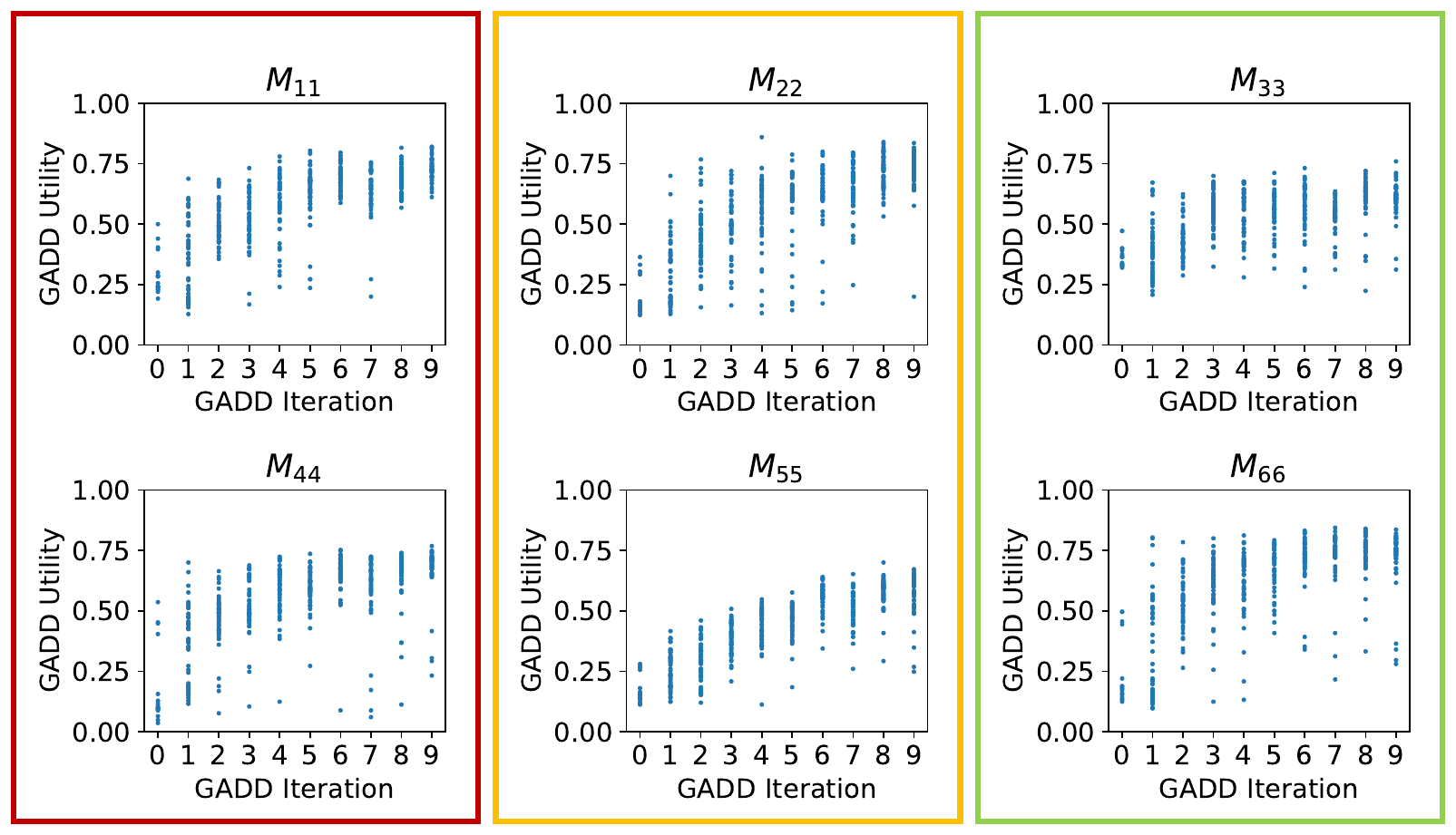}
    \caption{GADD learning utilities for each individual in the genetic algorithm population at each iteration for each of the 6 DD motifs identified from the 30-qubit QFT+M circuit. Motifs associated with the same colored box (red, yellow, green) are trained in parallel in the genetic algorithm. For each motif, the utility increases gradually until exhibiting qualitative saturation behavior, indicating convergence of the genetic algorithm to the largest achievable GADD utility among DD strategies belonging to the space of possible strategies determined by the GADD hyperparameters. We perform 10 iterations to observe such qualitative convergence in all 6 motifs.}

    \label{fig:qftm_training_results}
\end{figure*}

\begin{table}[h]
\begin{tabular}{c|c}
\hline
\textbf{DD Motif} & \textbf{Learned Sequences} \\
\hline
$M_{11}$ & \begin{tabular}{@{}c@{}} 
$X_m\ X_m\ Z_p\ Y_m\ I_m\ X_p\ Z_p\ Z_p$ \\
$Z_p\ I_m\ I_p\ X_m\ Z_p\ Y_m\ X_m\ Y_m$ 
\end{tabular} \\
\hline
$M_{22}$ & \begin{tabular}{@{}c@{}} 
$X_m\ I_m\ Y_p\ Z_p\ I_m\ X_m\ I_m\ X_m$ \\
$X_m\ I_p\ X_p\ Z_m\ I_m\ X_p\ X_m\ Z_p$ 
\end{tabular} \\
\hline
$M_{33}$ & \begin{tabular}{@{}c@{}} 
$Y_m\ Z_p\ Z_m\ Y_p\ X_p\ Y_p\ I_m\ Z_p$ \\
$X_m\ I_p\ I_m\ X_p\ I_p\ I_m\ X_p\ X_m$ 
\end{tabular} \\
\hline
$M_{44}$ & \begin{tabular}{@{}c@{}} 
$X_p\ I_p\ Y_m\ Z_m\ X_m\ I_m\ Y_p\ Z_p$ \\
$Z_p\ Y_m\ Z_m\ Y_p\ Z_p\ I_p\ X_m\ Y_m$ 
\end{tabular} \\
\hline
$M_{55}$ & \begin{tabular}{@{}c@{}} 
$Y_p\ Z_p\ Y_m\ I_p\ X_m\ X_p\ I_p\ Z_p$ \\
$Y_p\ Y_m\ I_p\ I_p\ Z_p\ Z_p\ Y_p\ Y_m$ 
\end{tabular} \\
\hline
$M_{66}$ & \begin{tabular}{@{}c@{}} 
$X_m\ Y_m\ I_m\ Z_m\ Y_m\ Z_m\ I_p\ X_p$ \\
$Y_m\ Z_m\ I_m\ Y_p\ Y_p\ I_m\ X_m\ I_m$ 
\end{tabular} \\
\hline
\end{tabular}
\caption{Utility-maximizing DD strategies for each DD motif that the GADD algorithm converges to. We select the largest utility individual in the final iteration plotted in Fig.~\ref{fig:qftm_training_results} as the learned sequences, which is representative of the sequences that the genetic algorithm converges to, as the top-performing GADD sequences are always preserved from one iteration to the next.}
\label{table:topgaddsequences}
\end{table}

\subsection{Hardware specifications}
\label{sec:hardware_specifications_appendix}
We report device metadata for data collected in Table~\ref{table:config}. All data were collected between 27-30 August 2024; we report metadata calibrated from \verb|ibm_kyiv| over those dates on relevant qubits involved in the experiments.

\begin{table}[h]
    \centering
    \begin{tabular}{|c|c|c|c|}
    \hline
        ~ & Min & Mean & Max \\
        \hline
        $T_1$ ($\mu$s) & 38.23 & 229.03 & 437.89 \\
        \hline
        $T_2$ ($\mu$s) & 16.41 & 132.11 & 403.63 \\
        \hline
        1QG error (\%) & 0.01 & 0.07 & 0.43 \\
        \hline
        2QG error (\%) & 0.78 & 2.41 & 6.49 \\
        \hline
        RO error (\%) & 0.11 & 1.14 & 15.46\\
        \hline
        1QG duration (ns) & 49.8 & 49.8 & 49.8\\
        \hline
        2QG duration (ns) & 561.8 & 568.1 & 583.1 \\
        \hline
    \end{tabular}
    \caption{Device specifications for \texttt{ibm\_kyiv} at the time of experiment execution. 1QG, 2QG, and RO denote one-qubit gate, two-qubit gate, and readout, respectively.}
\label{table:config}
\end{table}
\section{Scaling rules for dynamic circuit DD optimization}
\label{sec:scaling}
We claim that the GADD training time for a dynamic circuit split into $M$ motifs scales as
\begin{equation}
    \frac{M}{p}\cdot N_{it} \cdot T
\end{equation}
where the $M$ motifs are trained as $p$ motifs per parallel section, $N_{it}$ iterations are performed, and $T$ is time per GADD iteration, as $M/p$ represents the number of distinct GADD experiments performed. The number of iterations $N_{it}$ required for convergence is difficult to quantitatively characterize for heuristic methods such as the genetic algorithm we implement here due to the stochastic and population-based nature of the algorithm, as well as difficulty in fully understanding the utility landscape of all DD strategies and sensitivity to the various problem hyperparameters~\cite{Mit99, Gol89, Rud94, Eib03}. However, we naively expect that $N_{it}$ increases slowly with the number of degrees of freedom due to its ability to identify a principal direction of utility increase through executing reproduction between many distinct individuals per iteration. Here, the degrees of freedom represent the number of independent DD pulses being optimized, leading to the number of colors (i.e., number of unique sequences per strategy), size of the decoupling group, and number of pulses per sequence as factors contributing to $N_{it}$. We find that moderately small values of the latter two hyperparameters ($\abs{\mathcal{G}} = 8$ and $L = 8$) provide efficient optimization in Ref.~\cite{Ton24} and use such values here. For the number of DD colors, we find significant increases in the number of degrees of freedom lead to moderate slowing in convergence (as defined as number of iterations to reach a given utility threshold), as demonstrated in simulations in Appendix B of Ref.~\cite{Ton24}; furthermore, we find empirically that 10 iterations is sufficient in finding DD sequences that roughly reach the GADD utility maximum in our experiments shown in Fig.~\ref{fig:qftm_training_results}. Finally, the time per iteration $T$ is linear in the execution time per parallelized motif block, corresponding to the number of dynamic circuit layers optimized per motif, as well as the number of shots per circuit executed in the genetic algorithm, which relates to the signal-to-noise threshold of disambiguating DD strategies with different performance in the utility function. Here, we use motifs composed of 5 dynamic circuit layers and perform $250$ shots per circuit, corresponding to a maximum Bernoulli variance of $\frac{0.5^2}{250} = 10^{-3}$. Finally, $T$ also scales linearly with genetic algorithm population size; for which we find $K = 16$ to be a good hyperparameter. Further discussion on hyperparameter optimization is present in Ref.~\cite{Ton24} as well as in earlier simulations in Ref.~\cite{Qui13}.

With the aforementioned hyperparameters, we performed $N_{it} = 10$ GADD iterations for each of $M/p = 6/2 = 3$ parallelized experiments with time $T = 16.92$ seconds per iteration of quantum device usage time, for a total of 507.6 seconds. We note that the classical processing time per iteration required for the genetic algorithm to evaluate and generate new sequences is negligible compared to $T$. Such training overhead associated with GADD achieves a reduction in MCM error rates by a factor of two to three across benchmarked qubit pairs on the IBM Eagle device, enabling us to achieve over ten-fold increases in QFT+M $\fp$ and high signal-to-noise in the QFT+M of entangled quantum states. Such error rate reductions also lead to alleviations in overhead for methods in quantum error mitigation such as probabilistic error cancellation, where the cost to achieve a given error rate is exponential in the number of sampled quantum circuits~\cite{Tem17,Cai23}.

\section{MCM-RB}
\label{sec:mcmrb_results_appendix}
We perform mid-circuit measurement randomized benchmarking to characterize the behavior of empirically learned DD sequences using the \verb|mcm-rb| subroutine circuits presented in Ref.~\cite{Gov23}. We select a sample of circuit layers from each of the 6 circuit registers $R_i$ for which we performed a unique empirical learning experiment; such layers were selected to be representative of the frequency of measurement-induced Z-phase and collisional errors throughout the entire training circuit. MCM-RB experiments were performed in parallel for qubit registers $R_1, R_3, R_5$ and $R_2, R_4, R_6$ (Fig.~\ref{fig:dynamic_qft_mcmrb}a) respectively. Sixty MCM-RB random circuits with $l$ layers for $l = [2,4,6,8,10,12]$ were generated. At each circuit layer, three qubits corresponding to the measured qubits in sampled circuit layers for each $R_i$ in the group of three qubit registers experience MCMs. The feedforward time is padded subsequent to each measurement layer, followed by random Clifford operations on all unitary qubits. Feedforward time is added by inserting an idle period with length prescribed by the calibrated feedforward time in Ref.~\cite{Bau24} and padding sequences with feedforward compensation with the same pulse timings to replicate the behavior of such sequences in the MCM-RB application circuit. The same set of $60$ random circuits per depth were used for both sets of qubit registers, as well as for the experiments with and without DD. All experiments were executed on \verb|ibm_kyiv| immediately subsequent to training.

We present the raw MCM-RB results for each measured-unitary pair, obtained by marginalizing the outcome distribution over all measured/unitary qubits present in a given experiment, in Fig.~\ref{fig:mcmrbcf1} and Fig.~\ref{fig:mcmrbcf2} for the two parallel experiments respectively. Upon collecting these data, we observe two qualitatively distinct behaviors, as expected given the conclusions from Ref.~\cite{Gov23}. We find that most qubit pairs qualitatively exhibit a significant EPL decrease of the unitary qubit upon application of GADD sequences, while the measured qubit exhibits small EPL relative to the unitary qubit and does not significantly change upon addition of DD, corresponding to the qualitative behavior associated with the measurement-induced Z-phase error discussed in Ref.~\cite{Gov23}. On the other hand, in individual cases, both measured and unitary qubits have nontrivial and similar EPL, indicative of measurement-induced two-qubit error caused by frequency collision. In this case, the unitary qubits experience similar EPL between GADD and no DD, as well as the collision-constrained DD sequence (Appendix~\ref{sec:applying_dd_to_training_appendix}), which mostly consists of identity operations.

\begin{figure}[h]
    \centering
    \includegraphics[width=\linewidth]{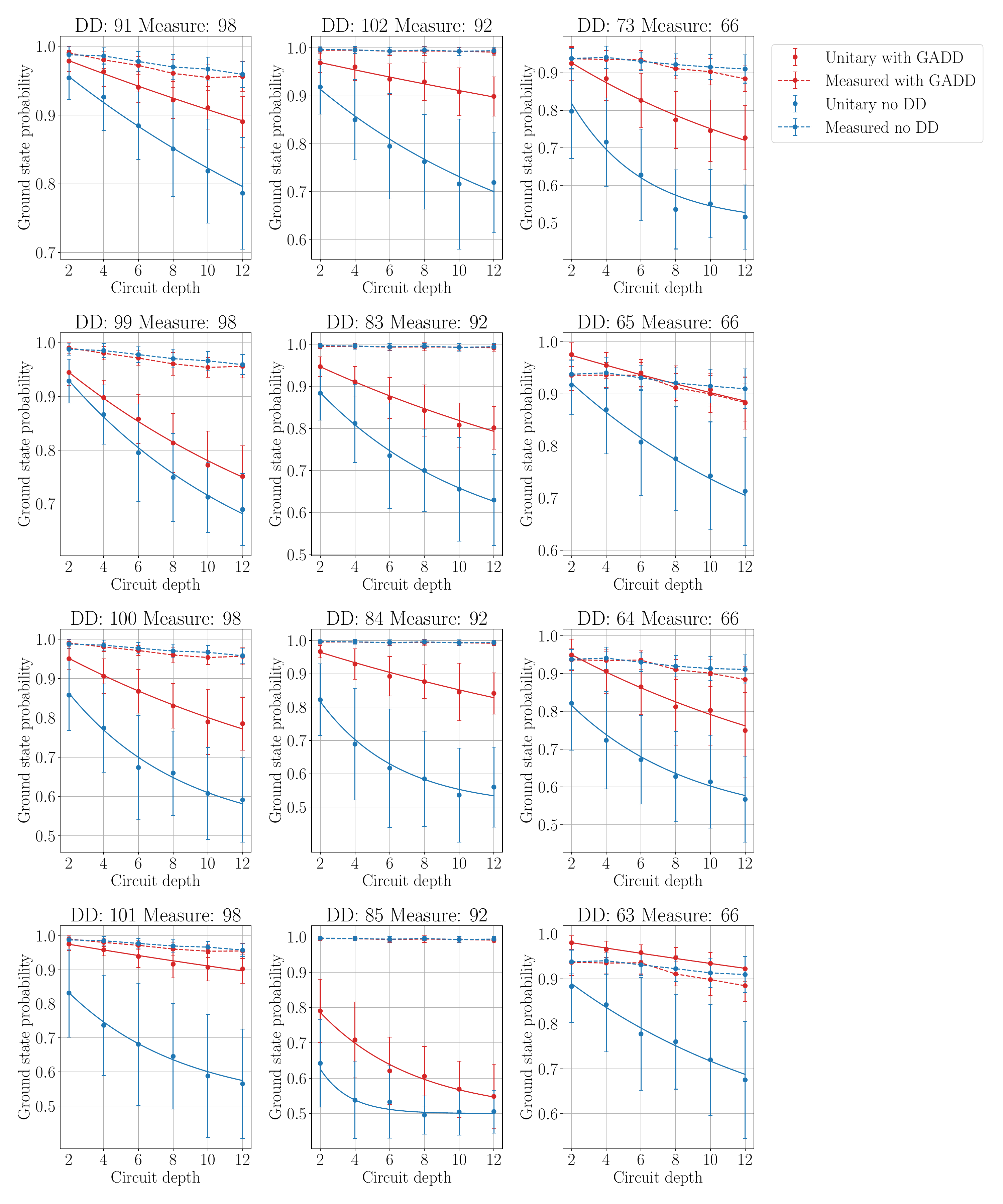}
    \caption{MCM-RB results on the first set of parallelized qubit registers, $R_1$, $R_3$, $R_5$, with (red) and without (blue) the application of empirically learned dynamical decoupling sequences. Statistics are derived from bootstrap resampling of $30$ random circuits from the ensemble of $60$ random circuits and error bars represent $1\sigma$ fluctuations in the mean $P(0)$ as calculated over resampled subensembles. MCM-RB data associated with unitary qubits are fit to an exponential decay for EPL calculation and the fit is shown in the solid line on each plot. The respective measured qubit data are not fit and are shown with the dashed line.}
    \label{fig:mcmrbcf1}
\end{figure}

\begin{figure}[h]
    \centering
    \includegraphics[width=\linewidth]{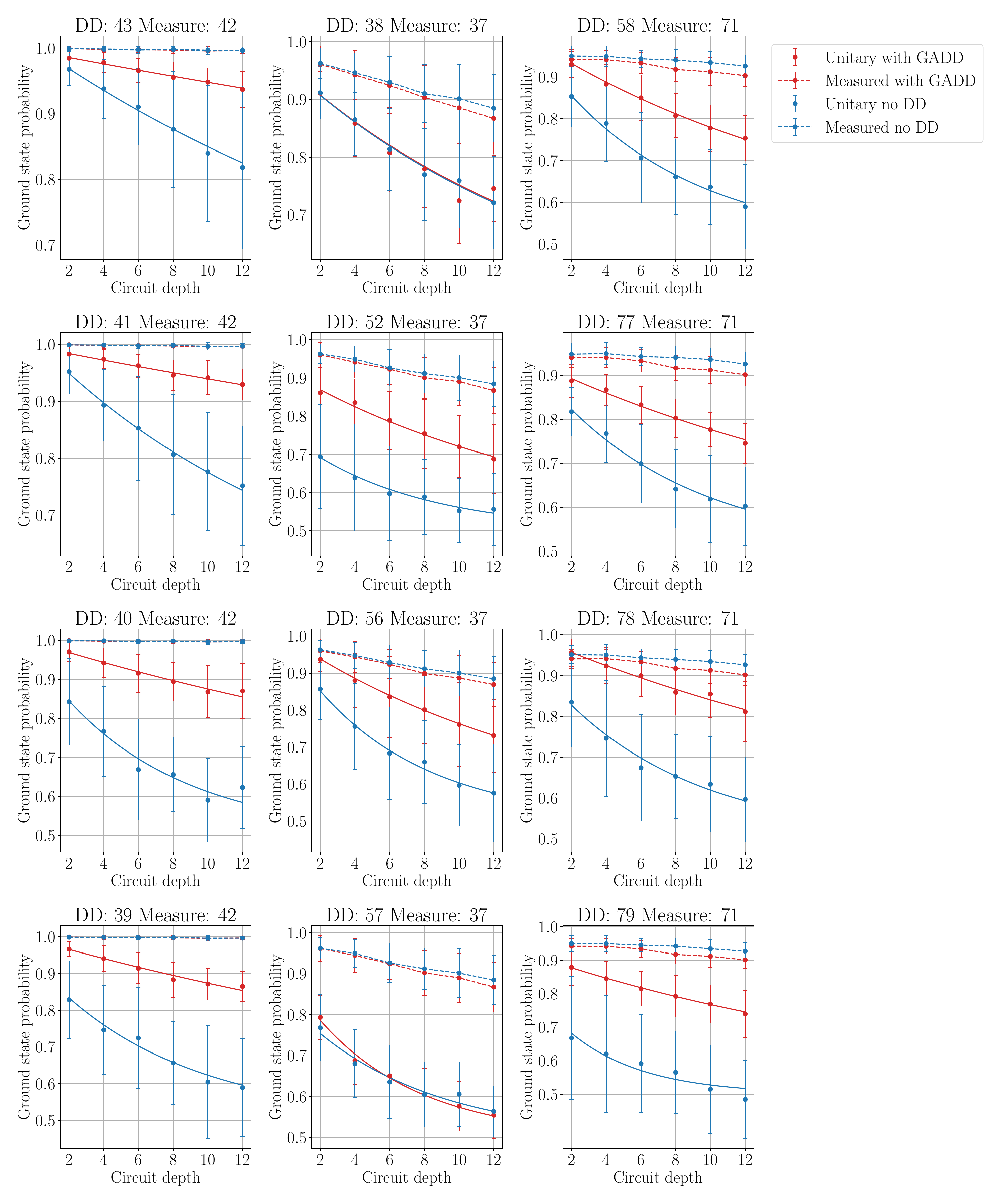}
    \caption{Results of MCM-RB experiments performed on the second set of parallelized qubit registers, $R_2$, $R_4$, $R_6$, analogous to those presented in Fig.~\ref{fig:mcmrbcf1}.}
    \label{fig:mcmrbcf2}
\end{figure}

\section{Dynamic circuits randomized benchmarking results}\label{sec:dcrb_results_appendix}

For DC-RB training, the utility metric used is the mean probability of measuring the unitary qubit in $\ket{0}$. After learning the best sequences using the motifs, we applied full DC-RB experiments with lengths $[0,1,2,3,4,5,10,15,20,35]$ and 7 randomizations per length; each circuit had 300 shots. 

We provide the full results for the data shown in Fig.~\ref{fig:dcrb_results}a and Fig.~\ref{fig:dcrb_results}b.

\begin{figure}[h]
    \centering
    \includegraphics[width=\linewidth]{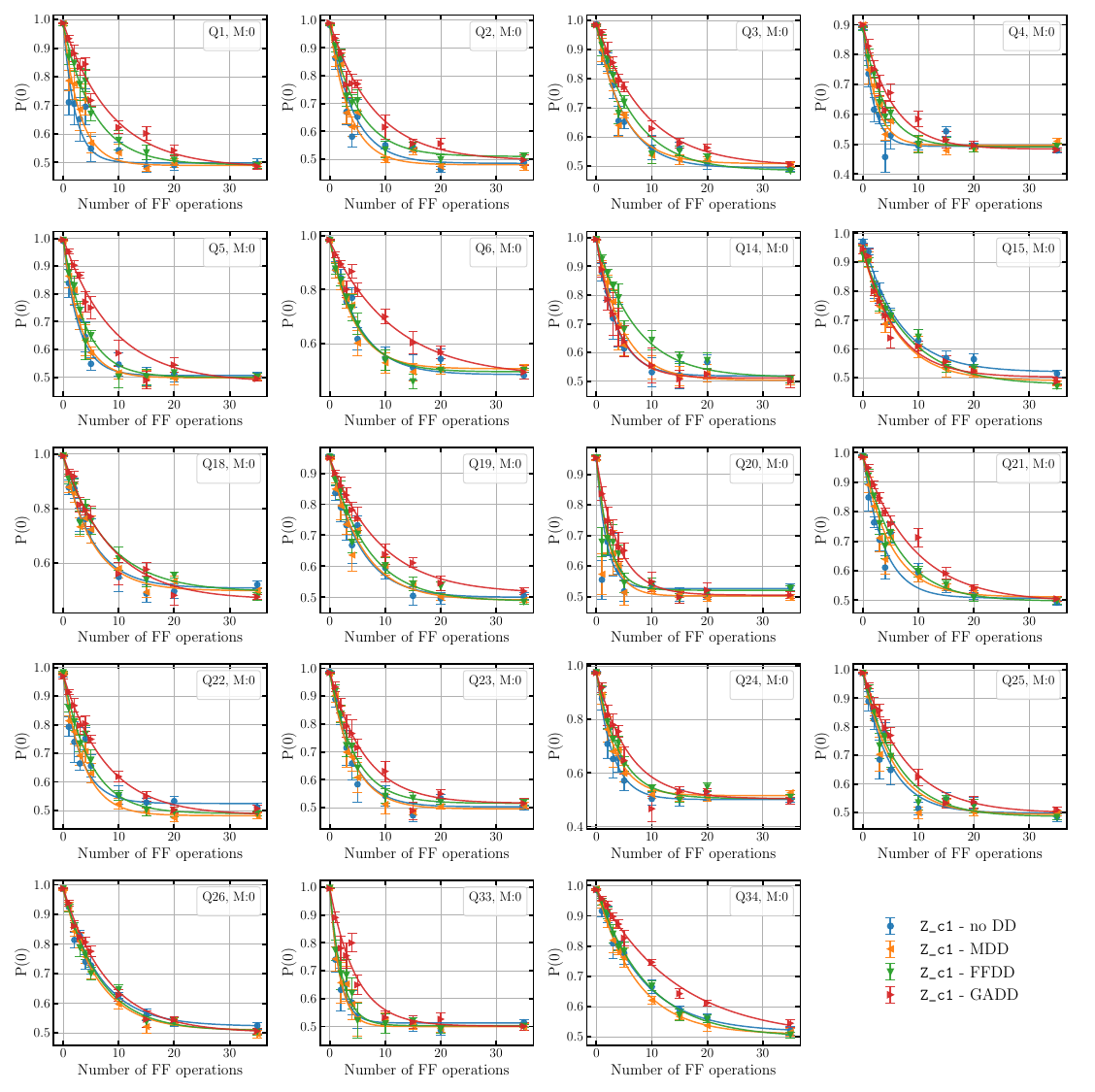}
    \cprotect\caption{Dynamic circuits randomized benchmarking results for the \verb|Z_c1| block. The extracted EPL are shown in Fig.~\ref{fig:dcrb_results}a. The protocol was applied to all qubits in parallel and the different DD sequences used the same 7 randomization circuits per length.}
    \label{fig:dcrb_curvesZ}
\end{figure}

\begin{figure}[h]
    \centering
    \includegraphics[width=\linewidth]{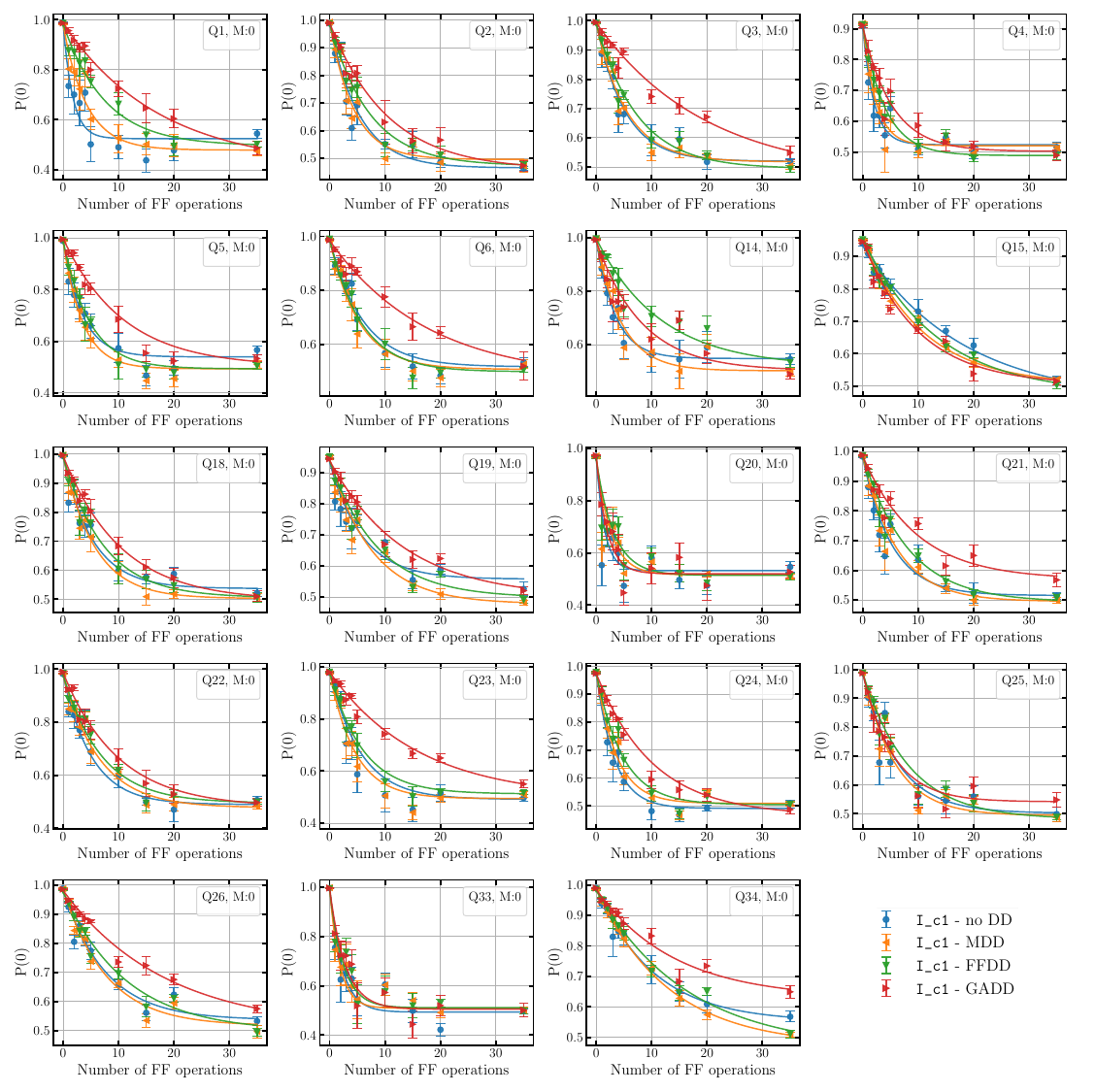}
    \cprotect\caption{Dynamic circuits randomized benchmarking results for the \verb|I_c1| block. The extracted EPL are shown in Fig.~\ref{fig:dcrb_results}b. The protocol was applied to all qubits in parallel and the different DD sequences used the same 7 randomization circuits per length.}
    \label{fig:dcrb_curvesI}
\end{figure}
\pagebreak
\section{QFT+M of X-basis GHZ states with local $R_z(\pi)$ rotations}
\label{sec:qft_on_ghz_full}
\begin{theorem}
    Let 
    \begin{equation}
        \begin{gathered}
        \ket{\Psi_m} = \frac{1}{2^{(N+1)/2}}\Bigg(\underbrace{(\ket{0} + \ket{1})\dots(\ket{0} + \ket{1})}_{N-m-1}(\ket{0} - \ket{1})\underbrace{(\ket{0} + \ket{1})\dots(\ket{0} + \ket{1})}_{m} \\ + \underbrace{(\ket{0} - \ket{1})\dots(\ket{0} - \ket{1})}_{N-m-1}(\ket{0} + \ket{1})\underbrace{(\ket{0} - \ket{1})\dots(\ket{0} - \ket{1})}_{m}\Bigg)
        \end{gathered}
    \end{equation}
    denote the $X$-basis GHZ state on $N$ qubits with $R_z(\pi)$ rotation on qubit $m$. We claim that 
    \begin{equation}
        0 \leq \abs{\bra{2^{N-m-1}}\text{QFT}\ket{\Psi_m}}^2 - \frac{2}{\pi^2} \leq \epsilon
    \end{equation}
for $\epsilon > 0, \epsilon \sim \mathcal{O}(2^{-2m})$.
\end{theorem}
\begin{proof}
    For any computational basis state $0\leq s \leq 2^N - 1$, let $w_i(s)$ denote the value of the bit at the $2^i$ place value in the binary representation of $s$ for $0 \leq i \leq N - 1$. Then, 
    \begin{equation}
        \ket{\Psi_m} = \frac{1}{2^{(N+1)/2}}\sum_{s = 0}^{2^N-1} (-1)^{w_m(s)}\ket{s} + (-1)^{W(s) - w_m(s)}\ket{s}
    \end{equation}
    where $W(s) = \sum_{i=0}^{N-1} w_i(s)$. Note $(-1)^{w_m(s)} = (-1)^{W - w_m(s)}$ if and only if $W(s) - 2w_m(s)$ is even; otherwise, $(-1)^{w_m(s)} + (-1)^{W(s) - w_m(s)} = 0$. Since $w_m(s)$ only takes on the values of $0$ and $1$, this occurs if and only if $W(s)$ is even. \\

    Let $\mathcal{S}$ denote the set of all $s$ with even $W(s)$. Note that, as there is a one-to-one correspondence between computational basis states with even $W(s)$ and odd $W(s)$ through flipping the least significant bit, $\abs{\mathcal{S}} = 2^{N-1}$. Thus, 
    \begin{equation}
        \ket{\Psi_m} = \frac{1}{2^{(N-1)/2}}\sum_{s\in\mathcal{S}}(-1)^{w_m(s)}\ket{s} 
    \end{equation}
    We aim to bound
    \begin{equation}
        \abs{\bra{2^{N-m-1}}\text{QFT}\ket{\Psi_m}}^2 = \frac{1}{2^{N-1}}\abs{\sum_{s\in\mathcal{S}}(-1)^{w_m(s)}\bra{2^{N-m-1}}\text{QFT}\ket{s}}^2
    \end{equation}
    The matrix element $\bra{2^{N-m-1}}\text{QFT}\ket{s}$ is $\frac{1}{2^{N/2}}e^{2\pi is\cdot2^{N-m-1}/2^N} = \frac{1}{2^{N/2}}e^{2\pi is/2^{m+1}}$. Therefore,
    \begin{equation}
        \frac{1}{2^{N-1}}\abs{\sum_{s\in\mathcal{S}}(-1)^{w_m(s)}\bra{2^{N-m-1}}\text{QFT}\ket{s}}^2 = \frac{1}{2^{2N-1}}\abs{\sum_{s\in\mathcal{S}}(-1)^{w_m(s)}e^{2\pi is/2^{m+1}}}^2
    \end{equation}
    We note that $(-1)^{w_m(s)} = (e^{\pi i})^{w_m(s)} = e^{2\pi i\cdot2^mw_m(s)/2^{m+1}}$. Therefore,
    \begin{equation}
        \frac{1}{2^{2N-1}}\abs{\sum_{s\in\mathcal{S}}(-1)^{w_m(s)}e^{2\pi is/2^{m+1}}}^2 = \frac{1}{2^{2N-1}}\abs{\sum_{s\in\mathcal{S}}e^{2\pi i(s + 2^mw_m(s))/2^{m+1}}}^2
    \end{equation}
    Let $s_m = (s + 2^mw_m(s)) \bmod{2^{m+1}}$ and note $e^{2\pi i(s + 2^mw_m(s))/2^{m+1}} = e^{2\pi i s_m/2^{m+1}}$. Furthermore, let $s_+ = s - s_m$. Then, for any fixed $s_+$, when $w_m(s) = 0$, all possible length $m$ bitstrings of a fixed parity (e.g., either even or odd $\sum_{i=0}^{m-1} w_{i}(s)$) are attained for $s_m$. In contrast, when $w_m(s) = 1$, $s_m = s - 2^m$, which attains all possible length $m$ bitstrings of the opposite parity. Therefore, 
    \begin{equation}
    \frac{1}{2^{2N-1}}\abs{\sum_{s\in\mathcal{S}}e^{2\pi i(s + 2^mw_m(s))/2^{m+1}}}^2 = \frac{1}{2^{2N-1}}\abs{\frac{2^{N-1}}{2^m}\sum_{s_m=0}^{2^m-1} e^{2\pi is_m/2^{m+1}}}^2 = \frac{1}{2^{2m+1}}\abs{\sum_{s_m=0}^{2^m-1} e^{2\pi is_m/2^{m+1}}}^2
    \end{equation}
    where we attain the prefactor of $\frac{\abs{\mathcal{S}}}{2^m}$ corresponding to the number of distinct values of $s_+$. Per the sum of a finite geometric series, \begin{equation}
         \frac{1}{2^{2m+1}}\abs{\sum_{s_m=0}^{2^m-1} e^{2\pi is_m/2^{m+1}}}^2 = \frac{4}{2^{2m+1}}\frac{1}{\abs{1-e^{2\pi i/2^{m+1}}}^2} = \frac{1}{2^{2m+1}}\csc^2{\frac{\pi}{2^{m+1}}}
    \end{equation} 
    Per the Laurent series expansion of $\csc^2{\theta}$ about $\theta = 0$,
    \begin{equation}
         \frac{1}{2^{2m+1}}\csc^2{\frac{\pi}{2^{m+1}}} = \frac{2}{\pi^2} + \mathcal{O}(2^{-2m})
    \end{equation}
    where all subdominant terms are positive and polynomial in $2^{-2m}$ as desired. 
\end{proof}
\pagebreak
\section{QFT on GHZ experimental results}
\label{sec:qft_ghz_expt_all_results}  
\begin{figure}[h!]
    \centering
    \includegraphics[width=0.9\linewidth]{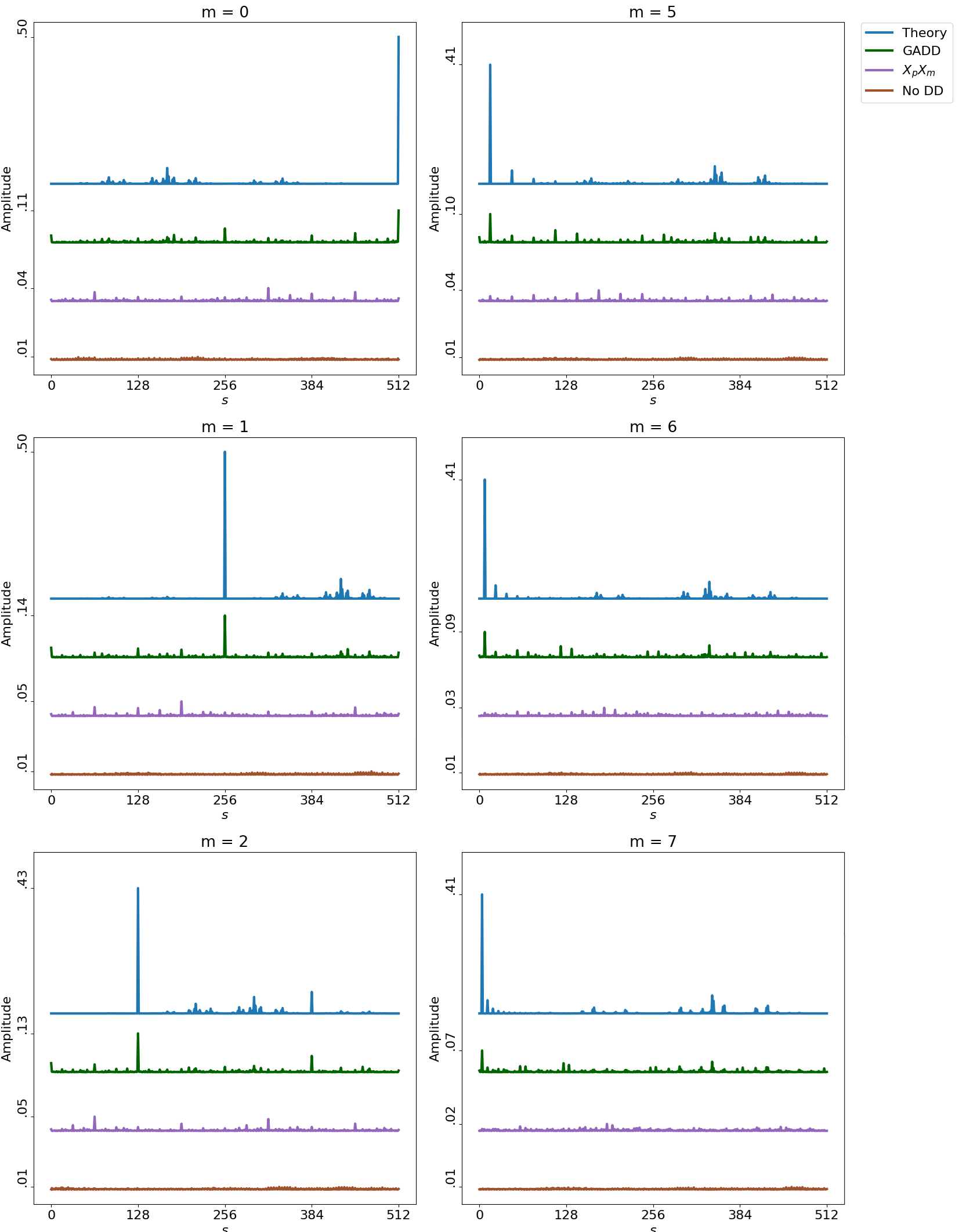}
    
\end{figure}

\renewcommand{\thefigure}{13}
\begin{figure}[t!]\ContinuedFloat
    \centering
    \includegraphics[width=0.8\linewidth]{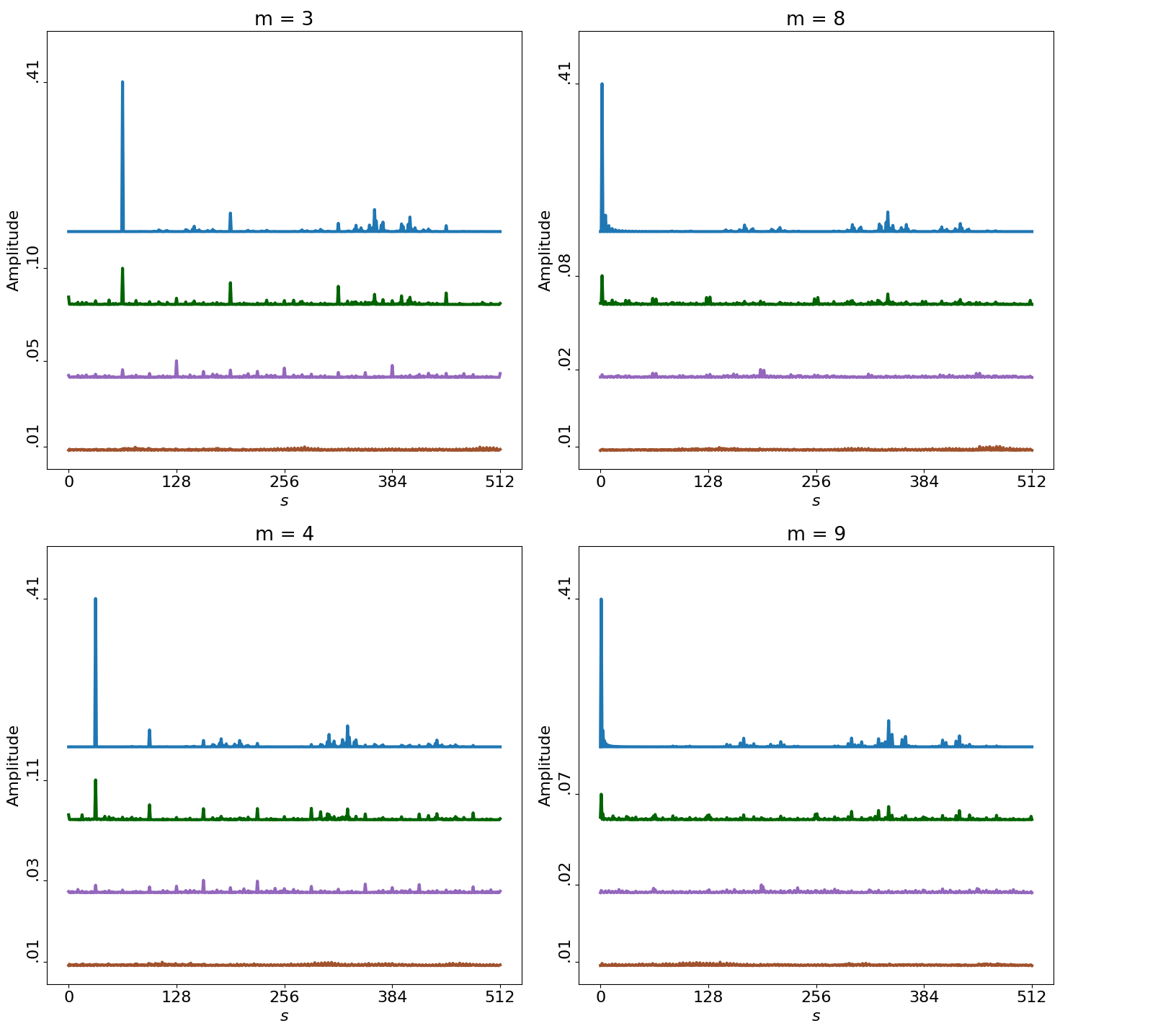}
    \caption{Raw counts measured over 10000 shots of the experiment described in Sec.~\ref{sec:qft_ghz_results} for computing $\text{QFT}\ket{\Psi_m}$ through the semiclassical QFT+M. We plot here the experimentally measured amplitude $\frac{1}{2^{10}}|(\bra{s} + \bra{2^{10} - s})\text{QFT}\ket{\Psi_m}|^2$ when GADD sequences are applied (green), compared to $X_pX_m$ staggered (purple) and no DD (brown) are applied. $X_pX_m$ staggered is implemented in the same uniform, FF-compensated fashion as the GADD sequences are applied. Experimental results are measured against the theoretically exact expected outcome distribution (blue).}
    \label{fig:qft_ghz_all_results} 
\end{figure}

\renewcommand{\thefigure}{14}
\begin{figure}
    \centering
    \includegraphics[width=0.5\linewidth]{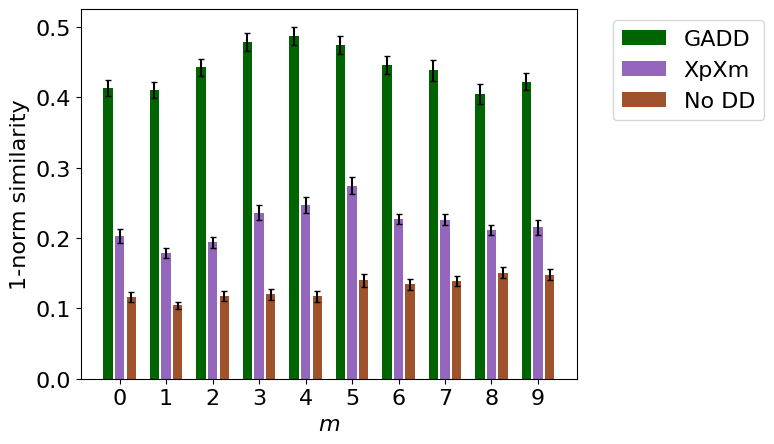}
    \caption{1-norm similarity (Eqn.~\ref{eqn:1norm_utility}) between each experimental data series (GADD, $X_pX_m$, and no DD) and the theoretical prediction; a value of 1 occurs if and only if the distributions are exactly the same, while a completely uniform distribution will achieve a value of near 0.}
    \label{fig:qft_onenorm_results}
\end{figure}

\end{document}